\documentclass[11pt]{article}

\usepackage{setspace,graphicx,amssymb,amsmath,mathtools,latexsym,amsfonts,amscd,amsthm,multirow,mathdots,caption,array,dsfont,wrapfig}
\usepackage{fancyhdr,tabularx,cite,mathrsfs,graphicx}
\usepackage{cite}
\usepackage{amsmath}
\usepackage{amssymb}
\usepackage{amsthm}
\usepackage[shortlabels]{enumitem}
\usepackage{authblk}
\usepackage{color,graphicx}
\usepackage{comment,mathtools}

\usepackage{tikz}
\usetikzlibrary{matrix}

\usepackage[linkcolor=blue]{hyperref}
\usepackage{geometry}
\renewcommand{\baselinestretch}{1.3}

\newcommand{\CC}{\mathbb{C}} 
\newcommand{\RR}{\mathbb{R}} 
 
\newcommand{\FF}{\mathbb{F}} 
\newcommand{\HH}{\mathbb{H}}
\newcommand{\ZZ}{\mathbb{Z}}

\newcommand{\bil}[2]{(#1,#2)}

\newcommand{\N}{\mathcal{N}}

\newcommand{\Co}{\textsl{Co}}

\newtheorem{theorem}{Theorem}
\newtheorem{proposition}[theorem]{Proposition}
\newtheorem{lemma}[theorem]{Lemma}

\newcommand{\SL}{\operatorname{\textsl{SL}}}      
\newcommand*{\mydprime}{^{\prime\prime}\mkern-1.2mu}

\DeclareMathOperator{\Aut}{Aut}

\DeclareMathOperator{\Tr}{Tr}
\newcommand{\ex}{\operatorname{e}}

\newcommand{\LL}{\Lambda}
\newcommand{\vfe}{V^f_{E_8}}

\newcommand{\vsn}{V^{s\natural}}

\newcommand{\vfn}{V^{f\natural}}

\newcommand{\cT}{{\cal T}}
\newcommand{\tw}{{\rm tw}}
\newcommand{\gz}{\mathfrak{z}}

\newcommand{\SO}{\operatorname{SO}}
\newcommand{\Spin}{\operatorname{Spin}}

\newcommand{\hookuparrow}{\mathrel{\rotatebox[origin=c]{90}{$\hookrightarrow$}}}

\numberwithin{equation}{section}

\def\be{\begin{equation}}
\def\ee{\end{equation}}

\allowdisplaybreaks[3]

\title{Vertex operator superalgebra/sigma model correspondences:\\ The four-torus case}
\author[1]{Vassilis Anagiannis\thanks{v.anagiannis@uva.nl}}
\author[1,2]{Miranda C. N. Cheng\thanks{mcheng@uva.nl}}
\author[3]{John Duncan\thanks{john.duncan@emory.edu}}
\author[4]{\\ Roberto Volpato\thanks{volpato@pd.infn.it}}
\date{}
\affil[1]{\small{Institute of Physics, University of Amsterdam, Amsterdam, the Netherlands}}
\affil[2]{\small{Korteweg-de Vries Institute for Mathematics, University of Amsterdam, Amsterdam, the Netherlands}}
\affil[3]{\small{Department of Mathematics, Emory University, Atlanta, GA 30322, U.S.A.}}
\affil[4]{\small{Dipartimento di Fisica e Astronomia `Galileo Galilei' e INFN sez. di Padova\authorcr Via Marzolo 8, 35131 Padova, Italy}}

\begin{document}

\maketitle

\abstract{
We propose a correspondence between vertex operator superalgebras and families of sigma models in which the two structures are related by symmetry properties and a certain reflection procedure. 
The existence of such a correspondence is motivated by previous work on ${\cal N}=(4,4)$ supersymmetric non-linear sigma models on K3 surfaces and on a vertex operator superalgebra with Conway group symmetry.
Here we present an example of the correspondence for ${\cal N}=(4,4)$ supersymmetric non-linear sigma models on four-tori, and compare it to the K3 case.

}

\clearpage

\tableofcontents

\newpage

{\centering{\em{{\bf\em In the memory of Prof. Tohru Eguchi}\\
As students of string theory and as curious mathematicians, we needed to study various papers of Professor Eguchi and his collaborators. A significant example is the review ``Gravitation, Gauge Theory and Differential Geometry" of almost 200 pages. As researchers, we have been seduced by moonshine phenomena for mock modular objects, the temptation for which must be blamed upon the paper "Notes on the K3 Surfaces and Mathieu Group $M_{24}$". We have been missing, and will continue to miss Eguchi-san and his inspiring work, as well as the unassuming, creative and curious manner in which he discussed and talked with us in person.
}}}

\section{Introduction}
\label{sec:intro}

The relation between sporadic finite simple groups and 
symmetries of K3 surfaces and K3 sigma models has attracted a lot of attention since the
pioneering work of \cite{mukai} and \cite{eot}. For some instances of this see \cite{eguchi1, eguchi2, eguchi3, dyons, gaberdiel_math1, gaberdiel_math2, eguchi4, surf1, umbral1, umbral2, surf2, surf3, LG, HM2020}. 
Apart from the Mathieu groups featured in \cite{mukai, eot}, 
symmetries of 
${\cal N}=(4,4)$ supersymmetric non-linear sigma models 
on K3 surfaces
have also been related to other groups, including the sporadic simple Conway groups \cite{gaberdiel_k3, huybrecht, derived}, and the groups of umbral moonshine \cite{umbral_k3, k3_lattices}.

The so-called twined elliptic genera play a critical role in quantifying this relation since they 
are sensitive to the way that symmetries act on quantum states. Of special interest is the fact that many of the twined elliptic genera of sigma models on K3 surfaces can be reproduced by the
vertex operator superalgebra (VOSA) $V^{s\natural}$, which has played a prominent role in Conway moonshine \cite{Duncan07,conway, derived}. 
(Here and in the remainder of this work we use {\em sigma model} as a shorthand for ${\cal N}=(4,4)$ supersymmetric non-linear sigma model.)

The analysis of \cite{k3_lattices} indicates that not all the twined K3 elliptic genera can be reproduced by Conway group symmetries of $V^{s\natural}$. It is nonetheless interesting that the single VOSA $\vsn$ can capture the symmetry properties of a large family of sigma models in the K3 moduli space, especially given that $V^{s\natural}$ is, in physical terms, a chiral theory, with central charge $c=12$, while the K3 sigma models 
are non-chiral theories, with $c=\bar{c}=6$. Moreover, in \S\ref{app:k3teg} we explain how all but one of the twined K3 elliptic genera  may be recovered from $\vsn$ if we allow non-Conway group symmetries (which is to say symmetries that do not preserve supersymmetry), or Conway group symmetries that are not of the expected order. 

This novel chiral/non-chiral connection between $\vsn$ and K3 sigma models has been made precise at a special (orbifold) point in the moduli space, where $V^{s\natural}$ can be retrieved as the image of the corresponding K3 theory under {\em reflection}: 
a procedure explored in \cite{derived} for the specific case of $V^{s\natural}$ and later formerly investigated in more generality by Taormina--Wendland in \cite{k3refl}. (See also \cite{refl} for a complementary approach). \par

To put this connection in a more structured context let us consider sigma models with target space
$X$ within one connected component of the full moduli space ${\cal M}={\cal M}(X)$ of sigma models on $X$, and denote
the corresponding sigma models by $\Sigma(X;\mu)$, for $\mu$ a
point in $\cal M$. For instance, for $X=T^4$ or $X=K3$ the moduli space consists of a single component,  
and takes the form 
\begin{equation}\label{def:modspace}
\begin{split}
&{\cal M}(T^4)= 
\left(SO(4)\times SO(4)\right)\backslash SO^+(4,4)/SO^+(\Gamma^{4,4})~,\\
&{\cal M}(K3)= \left(SO(4)\times O(20)\right)\backslash O^+(4,20)/ O^+(\Gamma^{4,20})~.
\end{split}
\end{equation}
Here $\Gamma^{a,b}$ denotes an even unimodular lattice of signature $(a,b)$. %

The chiral/non-chiral connection between $\vsn$ and 
K3 sigma models discussed above
now motivates the following question: 
\begin{quote}
{\em 
Are there 
pairs of 
VOSA/sigma model family pairs $(V,{\cal M}(X))$
such that 
the
following properties hold? 
\begin{enumerate}
\item The symmetry group of  
$V=V(X)$ contains the symmetry groups 
of all of the $\Sigma(X;\mu)$ for $\mu\subset{\cal M}(X)$. \label{property-sym}
\item The twined elliptic genera of ${V}$ 
\label{property-teg}
capture 
the twined elliptic genera arising from the $\Sigma(X;\mu)$ for all $\mu\in {\cal M}$.
\item There exists a particular point $\mu^*\in{\cal M}$ such that the reflection procedure 
maps $\Sigma(X;\mu^*)$ to ${V}$.\label{property-ref}
\end{enumerate}
}
\end{quote}
We will refer to pairs $(V,{\cal M})$ satisfying these 3 properties as 
{\em VOSA/sigma model correspondences}.

As we have explained, 
$(\vsn,{\cal M}(K3))$ comes tantalisingly close to being 
 an example of such a VOSA/sigma model correspondence. However, there are (conjecturally) a handful of twined elliptic genera of $\Sigma(X;\mu)$, with $\mu$ lying in certain high codimensional subspaces of ${\cal M}(X)$, that do not arise from $\vsn$. See Conjectures 5 and 6, and Table 4 of \cite{k3_lattices}. As a result, Property \ref{property-teg} above fails to hold for the $(\vsn,{\cal M}(K3))$ pair. Our main objective in this work is to illustrate a complete example of the correspondence, where K3 surfaces are replaced by (complex) four-dimensional tori.
The counterpart to $\vsn$ in this case is the 
VOSA naturally associated to the $E_8$ lattice, which we here denote $\vfe$ (as in \cite{Duncan07,FLMBerk}).
With the $K3$ case in mind this is perhaps unsurprising, given that $V^{s\natural}$ can be written as a suitable $\ZZ_2$ orbifold of $\vfe$ (see \cite{Duncan07,FLMBerk}), while on the orbifold locus of ${\cal M}(K3)$, the  
corresponding sigma models can also be obtained as $\ZZ_2$ orbifolds of four-torus sigma models (see Figure \ref{diagram}).
\begin{figure}
\begin{center}
\includegraphics[scale=0.35]{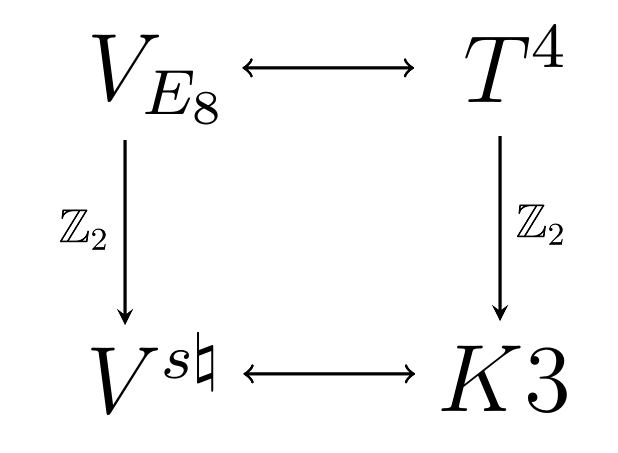}
\end{center}
\vspace*{-5mm}
\caption{VOSA/sigma model connections and the orbifold procedure.}
\label{fig:Z2diagram}
\end{figure}
In fact, as we will see, the VOSA/sigma model correspondence works better in the four-torus case since it holds for all points in ${\cal M}(T^4)$: The twined elliptic genera of \emph{any} $\Sigma(T^4;\mu)$ can be reproduced by the supersymmetry preserving twined elliptic genera of $\vfe$. (See Theorem \ref{thm:twining}.) 
So all three properties of our proposed VOSA/sigma model correspondence, including the one which failed for the $(\vsn, {\cal M}(K3))$ example, indeed hold in this case. It would be very interesting to understand whether a \emph{complete} realization of the VOSA/sigma model correspondence might exist even for K3 surfaces. Our results can be regarded as encouraging evidence in this direction.

The rest of the paper is organized as follows. In \S\ref{sec:t4model} we discuss the supersymmetry-preserving symmetries of $\Sigma(T^4;\mu)$ across the moduli space, as well as the corresponding twined elliptic genera. 
In \S\ref{sec:symmetry groups} we summarise important results on the groups arising in \S\ref{sec:t4model}. 
In \S\ref{sec:e8model} we discuss the VOSA $\vfe$, naturally associated to the $E_8$ lattice, and show that its supersymmetry-preserving symmetry group contains all the symmetry groups discussed in \S\ref{sec:t4model}. Hence we obtain that Property \ref{property-sym} of VOSA/sigma model correspondences holds for $(\vfe,{\cal M}(T^4))$. 
We then prove in Theorem \ref{thm:twining} that the VOSA $\vfe$ recovers
all the twined elliptic general of the $\Sigma(T^4;\mu)$, thereby proving 
Property \ref{property-teg}. 

In \S\ref{sec:orbifolds} we elaborate on the relation between the VOSA/sigma model correspondences for 
$T^4$ and the near example for $K3$ via orbifolding. In particular, we prove in Proposition \ref{pro:orbifold} that the diagram in Figure \ref{fig:Z2diagram} commutes, 
for all orbifolding procedures of the theory.
Then in \S\ref{sec:reflection} we demonstrate that $\vfe$ can be obtained as the image of  $\Sigma(T^4;\mu^\ast)$ at a particular special point $\mu^\ast\in {\cal M}(T^4)$ under reflection, thus establishing the final VOSA/sigma model correspondence property (Property \ref{property-ref}) for $(\vfe,{\cal M}(T^4))$. This is the content of Theorem \ref{thm:reflection}. 

We conclude the paper with three appendices. In the first of these, \S\ref{app:sym}, we provide further information on the supersymmetry-preserving symmetries of four-torus sigma models. 
In 
\S\ref{app:cocycles} we recall, for the convenience of the reader, how automorphisms of a lattice lift to automorphisms of a corresponding lattice VOSA, and detail the workings of this in the specific case of $\vfe$.
Finally, in 
\S\ref{app:k3teg} we explain how more general twinings of $\vsn$ may be used to recover the twined K3 elliptic genera that were not computed in \cite{derived}. We also review the relationship between $\vfn$ \cite{Duncan07} and $\vsn$ \cite{conway,derived}, 
explain a sense in which the Conway group arises naturally as a group of automorphisms of $\vsn$, 
and
explain why they are the same as far as twinings of the K3 elliptic genus are concerned.

\section{The Sigma Models} 
\label{sec:t4model}

In this section we setup our notations and collect important background on  four-torus sigma models  and their symmetries. The exposition follows closely that in \cite{t4paper}.

\subsection{Symmetries}
\label{subsec:sym_T4}

A sigma model on $T^4$ is a supersymmetric conformal field theory defined in terms of four pairs of left- and right-moving bosonic $u(1)$ currents $j^a(z),\tilde{j}^a(\bar{z})$, with $a=1,\dots,4$, four pairs of left- and right-moving free real fermions $\psi^a(z),\tilde{\psi}^a(\bar{z})$, as well as exponential (primary) fields $V_k(z,\bar{z})$ labelled by vectors $k=(k_L,k_R)\in\Gamma^{4,4}_{\rm w-m}$. 

Let us now explain our notation. Let $\Gamma^{4,4}$ denote an 
even  unimodular lattice of signature $(4,4)$. The real vector space
\begin{equation}
\Pi=\Gamma^{4,4} \otimes\RR\cong\RR^{4,4}
\end{equation} 
admits orthogonal decompositions into positive- and negative-definite subspaces
\be\label{orth_decomp_w-m}
\Pi=\Pi_L\oplus_{\perp}\Pi_R.
\ee
Correspondingly, we decompose $k\in \Pi$ as $k=(k_L,0)+(0,k_R)$, where the two summands lie in the positive- and negative-definite subspaces respectively. 
The relative position of $\Pi_L$ and $\Pi_R$ uniquely determines each four-torus sigma model, and the corresponding Narain moduli space is as in (\ref{def:modspace}),
where $O(\Gamma^{4,4})$ acts as $T$-dualities and we restrict to the $T$-dualities that moreover preserve world-sheet parity (cf. \cite{k3_lattices}). We use $\Gamma^{4,4}_{\rm w-m}$ to denote the lattice $\Gamma^{4,4}$ equipped with a choice of an orthogonal decomposition into positive- and negative-definite subspaces. This structure is also known as the winding-momentum or Narain lattice in this context.

The chiral algebra of every four-torus sigma model contains an  $\mathfrak{u}(1)^4$ algebra generated by the currents $j^a$, as well as an $\mathfrak{so}(4)_1$ Kac-Moody algebra generated by $:\psi^a\psi^b:$, with $a, b =1,\dots, 4$. It also contains a small $\mathcal{N}=(4,4)$ superconformal algebra at central charge $c=\tilde{c}=6$, whose holomorphic part is generated by the holomorphic stress tensor $T(z)$, four supercurrents $G^\pm(z),G^{\prime\pm}(z)$ of weight $(3/2,0)$ that consist of linear combinations of terms of the form $:\psi^aj^b:$. In particular, the  fermionic $\mathfrak{so}(4)_1$ algebra contains an $\mathfrak{su}(2)_1$ `R-symmetry' Kac-Moody algebra, generated by currents $J^1,J^2,J^3$. Since the anti-chiral discussion is completely analogous, from now on we focus just on the chiral part. 

To describe the superconformal algebra in detail, 
it is convenient to define complex fermions
\begin{equation}
\label{complex_fermions}
\begin{split}
\chi^1:=\frac{1}{\sqrt{2}}(\psi^1+i\psi^3)~,~~\chi^{1^*}:=\frac{1}{\sqrt{2}}(\psi^1-i\psi^3)~, \\
\chi^2:=\frac{1}{\sqrt{2}}(\psi^2+i\psi^4)~,~~\chi^{2^*}:=\frac{1}{\sqrt{2}}(\psi^2-i\psi^4)~,
\end{split}
\end{equation}
obeying the standard OPEs
\begin{equation}
\chi^i(z)\chi^j(w)\sim\mathcal{O}(z-w)~,~~\chi^i(z)\chi^{j^*}(w)\sim\chi^{i^*}(z)\chi^j(w)\sim\frac{\delta_{ij}}{z-w}~.
\end{equation}
In terms of the complex fermions, the stress tensor is given by
\begin{equation}
T=-\sum_{a=1}^4 :j^aj^a:-\frac{1}{2}\sum_{i=1}^2(:\chi^i\partial{\chi^{i}}^*:+:{\chi^{i}}^*\partial\chi^i:)~,
\end{equation}
while the R-symmetry currents are given by\footnote{Note that this normalisation for the currents, while convenient and common in the physics literature, differs by a factor of $\frac{1}{2}$ from the normalisation that is common in the Kac--Moody algebra context. }
\begin{equation}
\label{R_symmetry_currents}
\begin{split}
&J^1={i}\left(:\chi^1\chi^2:+:\chi^{1^*}\chi^{2^*}:\right)~,~~J^2=~:\chi^1\chi^2:-:\chi^{1^*}\chi^{2^*}:~, \\
&J^3=~:{\chi^1}\chi^{1^*}:+:\chi^2\chi^{2^*}:~. 
\end{split}
\end{equation}

The symmetry groups occuring at different points in the moduli space of sigma models on $T^4$ that preserve the $\mathcal{N}=(4,4)$ superconformal algebra were fully classified in \cite{t4paper}. 
To describe these groups, let $U(1)_{L}^4$ and $U(1)_R^4$ be the Lie groups generated by the zero modes $j_0^a$ and $\tilde{j}_0^a$ respectively.
They describe the (independent) translations along the four-torus. 
Recall also that apart from the R-symmetry $\mathfrak{su}(2)_1$ algebra with generators \eqref{R_symmetry_currents}, there is another copy of $\mathfrak{su}(2)_1$ algebra in the fermionic $\mathfrak{so}(4)_1$ algebra, generated by the currents
 \begin{equation}\label{def:sigma_su2A}
\begin{split}
&A^1=i\left(:\chi^1\chi^{2^*}:+:\chi^{1^*}\chi^2:\right)~,~~A^2=~:\chi^1\chi^{2^*}:-:\chi^{1^*}\chi^2:~, \\
&A^3=~:{\chi^1}\chi^{1^*}:-:\chi^2\chi^{2^*}:~.
\end{split}
\end{equation}
Focussing on the zero modes, we have the relation
\be
SO(4)_L \cong (SU(2)_L^J\times SU(2)_L^A)/(-1)^{A_0^3+J_0^3},
\ee
where $(-1)^{A_0^3/J_0^3}$  is the non-trivial central element of $SU(2)_L^{A/J}$, 
and similarly for the right-moving side. 
Preserving the ${\cal N}=4$ superconformal algebra restricts us to the subgroup  $SU(2)_L^A$ which commutes with the R-symmetry  $SU(2)_L^J$. Moreover, identifying $SO(4)_L$ with $SO(\Pi_L)$, we need to consider subgroups that induce an automorphism of $\Gamma^{4,4}_{\rm w-m}$\footnote{The identification between $SO(4)_L$ with $SO(\Pi_L)$ is  given by the choice of the ${\cal N}=1$ supercurrent such that its generator is proportional to $\sum_{a=1}^4 :\psi^aj^a:$. Different choices of the $\N=1$ supercharge lead to different isomorphisms that are related to each other by R-symmetry transformations in $SU(2)^J_L$. }.

These considerations lead to the following specification of the symmetry groups of the four-torus sigma models.
They take the form
\be 
\label{G_symmetries}
G=(U(1)_L^4\times U(1)_R^4).G_0\ .
\ee
The group $G_0$ here is given by the intersection
\be \label{def:G0}
G_0=\left(SU(2)_L^A\times SU(2)_R^A\right)\cap O\left(\Gamma^{4,4}_{\textrm{w--m}}\right)\ ,
\ee
where the above identification is understood.

Notice that the groups $G_0$ defined in \eqref{def:G0} manifestly do not mix the spaces $\Pi_L$ and $\Pi_R$, and always 
contains a central $\ZZ_2$ subgroup generated by $(-1,-1)\in SU(2)_L^A\times SU(2)_R^A$.
Consider the set of all possible groups arising as
\be
\label{G1G0}
G_1:=G_0/(-1,-1).
\ee
This set turns out to be bijective to the set of 
 subgroups of the group of even-determinant Weyl transformations of $E_8$, denoted by $W^+(E_8)$, that fix an $E_8$-sublattice of rank at least $4$. 
 See \cite{t4paper} for a complete and descriptive list of all the possible groups $G_0$. We note here that 
 the groups $G_0$ and $G_1$ are interesting finite groups only at certain special points in the moduli space ${\cal M}(T^4)$ of sigma models on $T^4$. Generically, $G_0$ is isomorphic to $\ZZ_2$ and $G_1$ is trivial.

\subsection{Twined Genera}
\label{sec:twinedsigma}
The elliptic genus of an $\mathcal{N}=(4,4)$ superconformal theory is defined in terms of the superconformal algebra generators as the following trace over the RR sector,
\begin{equation}\label{def:EG}
\phi(\tau,z)=\Tr_{\textrm{RR}}\left[(-1)^{F}y^{J_0^3} ~q^{L_0-\frac{c}{24}}~\bar{q}^{\tilde{L}_0-\frac{\tilde{c}}{24}}\right]~,~~q:=e^{2\pi i\tau}~,~~y:=e^{2\pi iz}~,
\end{equation}
where $L_0$ is the zero mode of the stress energy tensor $T$, and the fermion number operator $(-1)^F$ will be discussed in more detail later. It receives non-vanishing contributions only from right-moving BPS states and thus does not depend on $\bar{\tau}$. For the $\mathcal{N}=(4,4)$ theories that we are considering, it is also a  weak Jacobi form of weight $0$ and index $1$, and does not depend on the moduli. For four-torus sigma models, we have  $c=\tilde{c}=6$ and the elliptic genus is in fact identically zero due to cancelling contributions from the BPS states, which form an even-dimensional representation of the Clifford algebra of the right-moving fermionic zero modes $\tilde{\chi}^i_0, \tilde{\chi}^{i^*}_0$. When the theory has additional symmetries $G$ preserving the superconformal algebra (i.e. at special points in the moduli space), we can also consider the  elliptic genus twined by an element $g\in G$ acting on the RR states,
\begin{equation}
\phi^{G}_g(\tau,z)=\Tr_{\textrm{RR}}\left[g~(-1)^{F}~y^{J_0^3}~q^{L_0-\frac{c}{24}}~\bar{q}^{\tilde{L}_0-\frac{\tilde{c}}{24}}\right]\ ,
\end{equation}
where the upper-script in the notation serves to remind us about moduli dependence (through the symmetry group $G$). The twined genus $\phi^G_g$ depends only on the conjugacy class of $g$ in $G$ and is a weak Jacobi form of weight $0$ and index $1$ for some congruence subgroup $\Gamma_g \subseteq \SL_2(\ZZ)$. 
Note that the normal subgroup $U(1)_L^4\times U(1)_R^4$ of $G$ \eqref{G_symmetries} acts trivially on all oscillators. For this reason we will first focus on the $G_0$ part when computing the twined elliptic genera.

To compute the elliptic genus twined by $g\in G_0\subset SU(2)_L^A \times SU(2)_R^A$, 
let us first describe the Fock space representation of the RR states in the present theory. This is built from all possible combinations of the free fermionic  $\chi_n^i$, $\chi_n^{i*}$, $\tilde{\chi}_n^i$, $\tilde{\chi}^{i*}_n$ and bosonic oscillators $j^a_n$, $\tilde{j}^a_n$, with $a=1,\dots,4$, $i=1,2$ and $n\in \ZZ_{\leq -1}$, acting on the Fock space ground states. The latter has a  convenient basis given by 
\begin{equation}\label{eqn:gndstates}
|k_L,k_R;s\rangle~,~~s=(s_1,s_2;\tilde{s}_1,\tilde{s}_2)~,~~s_1,s_2,\tilde{s}_1,\tilde{s}_2\in\left\lbrace\frac{1}{2},-\frac{1}{2}\right\rbrace~.
\end{equation}
Here $s$ is an index for the $2^4$-dimensional representation of the eight-dimensional Clifford algebra generated by 
  the fermionic zero modes $\chi_0^i$, $\chi_0^{i*}$, $\tilde{\chi}_0^i$, $\tilde{\chi}^{i*}_0$, which correspond to the fermionic RR ground states $|s\rangle:=|0,0;s\rangle$. The indices $k_L$ and  $k_R$ label points in the winding-momentum lattice, $k=(k_L,k_R)\in\Gamma^{4,4}_{\textrm{w--m}}$. In terms of the primary operators $V_k(z,\bar{z})$, the ground states in \eqref{eqn:gndstates} are given by $|k_L,k_R;s\rangle:=V_k(0,0)|s\rangle$.

In this basis, the eigenvalues of the fermionic ground states under the operators $J_0^3$ and $\tilde{J}_0^3$ are given by
\begin{equation}
J_0^3|s\rangle=(s_1+s_2)|s\rangle~,~~\tilde{J}_0^3|s\rangle=(\tilde{s}_1+\tilde{s}_2)|s\rangle~,
\end{equation}
and similarly 
\begin{equation}
A_0^3|s\rangle=(s_1-s_2)|s\rangle~,~~\tilde{A}_0^3|s\rangle=(\tilde{s}_1-\tilde{s}_2)|s\rangle~,
\end{equation}
while the $J^3$ charges of the fields are  given by
\be
\begin{array}{c|c|c}
\chi^i & \chi_n^{i*} & j^a_n \\ \hline +1 & -1 & 0
\end{array}
\ee
and similarly for the right-movers.
In these terms, the fermion number operator is defined as $(-1)^F:=(-1)^{J_0^3+\tilde{J}_0^3}$.

Let $\rho_\psi$ denote the $8$-dimensional representation of $G_0$ on the space spanned by $\psi^1,\ldots,\psi^4$ and $\tilde\psi^1,\ldots,\tilde\psi^4$. For a given element $g\in G_0$, choose the parametrisation of the complex fermions  such that $g$ acts as (cf. Table \ref{t:classes})
\be\label{g_action}
\rho_\psi(g)\chi^1=\zeta_L\chi^1,\qquad\rho_\psi(g)\tilde \chi^1=\zeta_R\tilde\chi^1. 
\ee
Since $g\in SU(2)_L^A \times SU(2)_R^A$, it follows that $g$ acts on the eight-dimensional representation $\rho_\psi$ as
\begin{equation}\label{def:gact_chi}
\begin{split}
& \rho_\psi(g)\chi^1=\zeta_L\chi^1~,~~\rho_\psi(g)\chi^{1^*}=\zeta_L^{-1}\chi^{1^*}~,~~\rho_\psi(g)\tilde{\chi}^1=\zeta_R\tilde{\chi}^1~,~~\rho_\psi(g)\tilde{\chi}^{1^*}=\zeta_R^{-1}\tilde{\chi}^{1^*} \\
& \rho_\psi(g)\chi^2=\zeta_L^{-1}\chi^2~,~~\rho_\psi(g)\chi^{2^*}=\zeta_L\chi^{2^*}~,~~\rho_\psi(g)\tilde{\chi}^2=\zeta_R^{-1}\tilde{\chi}^2~,~~\rho_\psi(g)\tilde{\chi}^{2^*}=\zeta_R\tilde{\chi}^{2^*}~,
\end{split}
\end{equation}
and similarly on the bosonic currents since the superconformal algebra is preserved. Note that the choice of parametrisation in \eqref{g_action}  is always possible, since by conjugations in  $SU(2)_L^A \times SU(2)_R^A$ we can let $g$ to be contained in the Cartan subgroup generated by $A_0^3$ and  ${\tilde A}_0^3$.

From the preceding discussion we conclude that the twined elliptic genus of the four-torus sigma model factors as 
\be
\label{eigenvalues_vector}
\phi^G_g(\tau,z)=\phi_g^{\textrm{osc}}(\tau,z)\phi_g^{\textrm{gs}}(z)\phi_g^{\textrm{w--m}}(\tau)\ ,
\ee 
where the three factors capture the contributions from 
 the oscillators, 
 the fermionic ground states, and 
winding-momentum (i.e. primaries $V_k$),  respectively. In what follows we will discuss them separately.

The action on the ground states is given by
\begin{equation}
g|s\rangle=\zeta_L^{A_0^3}\zeta_R^{\tilde{A}_0^3}|s\rangle
=\zeta_L^{s_1-s_2}\zeta_R^{\tilde{s}_1-\tilde{s}_2}|s\rangle~.
\end{equation}
Summing over the $2^4$ ground states $|s\rangle$ we hence arrive at
\begin{gather}
\begin{split}
\phi^{\textrm{gs}}_g(z)&=y^{-1}(1-\zeta_L y)(1-\zeta^{-1}_L y)(1-\zeta_R)(1-\zeta^{-1}_R)\\&=2(1-\Re(\zeta_R))(y^{-1}+y-2\Re(\zeta_L))\ .
\end{split}
\end{gather}
From \eqref{def:gact_chi}, we compute that the total contribution from the fermionic and bosonic oscillators is
\small
\begin{gather}
\begin{split}
\phi^{\textrm{osc}}_g(\tau,z)
&=\prod_{n=1}^\infty \frac{(1-\zeta_L yq^n)(1-\zeta^{-1}_L yq^n)(1-\zeta_L y^{-1}q^n)(1-\zeta^{-1}_L y^{-1}q^n)}{(1-\zeta_L q^n)^2(1-\zeta^{-1}_L q^n)^2}~.
\end{split}
\end{gather}
Notice that the contribution from the right-moving oscillators, and thus the $\bar{\tau}$ dependence, cancels out completely.

Finally, the contribution from winding-momentum is given by
\be\label{w-m_phases}
\phi_g^{\textrm{w--m}}(\tau)=\sum_{k=\left(k_L,k_R\right)\in \left(\Gamma^{4,4}_\textrm{w--m}\right)^{g}} \xi_g(k_L,k_R)~q^{\frac{{k}_L^2}{2}}\, \bar q^{\frac{{k}_R^2}{2}}\ .
\ee
Here $\left(\Gamma^{4,4}_\textrm{w--m}\right)^{g}$ is the $g$-fixed sublattice of $\Gamma^{4,4}_\textrm{w--m}$, and $\xi_g\left(k_L,k_R\right)$ are suitable phases that depend on the choice of the lift of $g$ from $G_0$ to $G$. As discussed in \S\ref{app:cocycles} one can always choose the standard lift, where the phases $\xi_g(k_L,k_R)$ are trivial for all $(k_L,k_R)\in \left(\Gamma^{4,4}_\textrm{w--m}\right)^{g}$.

Notice that if $g$ acts trivially on the right-movers, then $\zeta_R=1$ and $\phi^{\rm gs}_g$, and therefore $\phi^G_g$ vanishes. 
On the other hand, if both $\zeta_R$ and $\zeta_L$ are different from one, then $\left(\Gamma^{4,4}_\textrm{w--m}\right)^{g}=\{0\}$ and $\phi_g^{\textrm{w--m}}=1$. 
Thus, determining $\phi_g^{\text{w-m}}$ is nontrivial only when $\zeta_R\neq1$ and $\zeta_L=1$. 
As a result, we can rewrite 
\be\label{w-m_phases2}
\phi_g^{\textrm{w--m}}(\tau)=\sum_{k=\left(k_L,0\right)\in \left(\Gamma^{4,4}_\textrm{w--m}\right)^{g}} \xi_g(k_L,0)~q^{\frac{{k}_L^2}{2}}
\ee
which is indeed holomorphic in $\tau$ as required.

\section{The Symmetry Groups}
\label{sec:symmetry groups}

In this section we establish notation and summarise important results on the groups that we will make use of later. 
In particular, we will show that the $G_0$, related to the total symmetry groups of the four-torus sigma models via \eqref{G_symmetries}, are all subgroups of $W^+(E_8)$, the group of even-determinant Weyl transformations of $E_8$. This fact will be crucial  in \S\ref{sec:e8model}, as it makes it possible to equate the twined elliptic genera of the four-torus sigma models and the twined traces of the $E_8$ lattice VOSA.

By definition, $W^+(E_8)$ has a natural action on the $E_8$ lattice via its unique eight-dimensional irreducible representation and is a subgroup of $SO(8)$. 
Under the inclusion map $W^+(E_8)\xhookrightarrow{} SO(8)$, the center of $W^+(E_8)$ is mapped to the central $\ZZ_2$ subgroup of $SO(8)$, acting as $-id$ in the  eight-dimensional vector representation of $SO(8)$ in the former case and in the eight-dimensional non-trivial representation of $W^+(E_8)$ in the latter case. 
We denote by $\iota_v$ the generator of this latter central subgroup $\langle\iota_v\rangle\cong\ZZ_2< W^+(E_8)$. The corresponding central quotient is isomorphic to the finite simple group $O^+_8(2)$, the group of linear transformations of the vector space $\FF_2^8$ preserving a certain quadratic form. (See e.g. \cite{ATLAS} for a discussion of this.) 
In other words, we have  $$W^+(E_8)\cong\langle \iota_v\rangle.O^+_8(2) ~. $$

Recall that $G_1$, related to $G_0$ as in \eqref{G1G0},  can be identified with  subgroups of $W^+(E_8)$ that fix an $E_8$ sublattice of rank at least $4$ \cite{t4paper}.
Since $\iota_v$ does not preserve any subspace in the eight-dimensional vector representation of $W^+(E_8)$, we conclude that $\iota_v\not \in G_1$, and by combining the inclusion $G_1\xhookrightarrow{} W^+(E_8)$  and the projection $W^+(E_8)\xrightarrow[]{\pi'} O^+_8(2)$ we obtain an injective homomorphism $G_1\to  O^+_8(2)$. As a consequence, the group  $G_1$ is always  isomorphic to a subgroup of $O^+_8(2)$.

To show that the discrete part of the sigma model symmetry group $G_0$ is always a subgroup of $W^+(E_8)$, it will be useful to consider the group  $\textrm{Spin}(8)$. The kernel of the spin covering map $\textrm{Spin}(8)\xrightarrow[]{\pi} SO(8)$ is an involution $\langle \iota_s\rangle\cong\ZZ_2$.
Considering $W^+(E_8)<SO(8)$, the preimage of the spin covering map is $\langle\iota_s\rangle.W^+(E_8)< \textrm{Spin}(8)$.  Its center can be  identified with the center of $\textrm{Spin}(8)$, given by $\langle\iota_s,\iota_v\rangle\cong\ZZ_2\times\ZZ_2$. We thus have that $$\langle\iota_s\rangle.W^+(E_8)\cong\langle\iota_s,\iota_v\rangle.O_8^+(2)~.$$

The kernel of the spin covering map $\textrm{Spin}(8)\xrightarrow[]{\pi} SO(8)$ is naturally identified with the kernel of the quotient map $G_0 \to G_1$  (cf. \eqref{G1G0}, Table \ref{diagram}). 
Indeed, the preimage of $G_1< W^+(E_8)< SO(8)$ in $\langle \iota_s\rangle.W^+(E_8)< \textrm{Spin}(8)$ is precisely the group $G_0\cong \langle\iota_s\rangle.G_1$. As we have seen in \S\ref{sec:twinedsigma}, in the sigma models $\iota_s$  acts by flipping the sign of all the fermions in the representation $\rho_\psi$ (cf. \eqref{def:gact_chi}).

At this point it is crucial to recall that $\textrm{Spin}(8)$ has a triality symmetry, i.e. an $S_3$ outer automorphism group. 
Also, it has 
one vector and the two spinor eight-dimensional irreducible representations, which we will denote by $\rho^s_\psi$, $\rho^s_{\rm e}$ and $\rho^s_{\rm o}$ respectively, and the action of triality on the group $\textrm{Spin}(8)$ extends to an  $S_3$ permutation action on the three representations $\rho^s_\psi$, $\rho^s_{\rm e}$ and $\rho^s_{\rm o}$.  This $S_3$ group also permutes the three non-trivial generators $\iota_v$, $\iota_s$, $\iota_v\iota_s$ of the center of $\textrm{Spin}(8)$, and in each of the three aforementioned eight-dimensional representations one of these generators acts trivially. Triality for $\textrm{Spin}(8)$ induces an $S_3$ group of outer automorphisms  of  $\langle\iota_s\rangle.W^+(E_8)\cong\langle\iota_s,\iota_v\rangle.O^+_8(2)$. 

As a result, the $G_0$ subgroup of $\langle\iota_s\rangle.W^+(E_8)$ has three representations, which we denote $\rho_\psi$, $\rho_{\rm e}$ and $\rho_{\rm o}$, corresponding to three eight-dimensional representations of $\textrm{Spin}(8)$,  that are permuted by the outer automorphisms of $\langle\iota_s\rangle.W^+(E_8)$.
As we have seen in \eqref{def:gact_chi}, in the sigma model the representation $\rho_\psi$ captures the action of the symmetry group $G_0$ on  the eight (left- and right-moving) NS-NS fermions $\chi^i, \chi^{i^*} \tilde{\chi}^i,\tilde{\chi}^{i^*}$. The other two representations, $\rho_{\rm e}$ resp. $\rho_{\rm o}$,  capture the action of $G_0$ on the Ramond-Ramond sector quantum states with even resp. odd fermion numbers. 
As mentioned before, in the representation $\rho_\psi$ the central involution $ \iota_s$ acts by flipping the signs of all fermions as well as all bosons (which has to be the case since $G_0$ preserves the superconformal algebra).
 On the other hand, in the representation $\rho_{\rm e}$ the central element of $G_0$ acts trivially, so that only the quotient $G_1$ acts faithfully on the RR ground states of even fermion numbers. This is also the representation where $G_1$ fixes a $4$-dimensional subspace (cf. Table \ref{t:classes}).

 Now the  $S_3$ outer automorphisms of $\langle\iota_s,\iota_v\rangle.O^+_8(2)$ guarantee that the quotient by any of the three generators of the central subgroup $\langle\iota_s,\iota_v\rangle$ is a group isomorphic to $W^+(E_8)$. 
 In particular, since $\iota_v \not\in G_1$ and hence $G_0\cong \langle\iota_s\rangle.G_1 < \langle \iota_s\rangle.W^+(E_8)$ does not contain the central involution $\iota_v$, the homomorphism $G_0\to W^+(E_8)$ induced by the projection
\begin{equation}
\langle \iota_s\rangle.W^+(E_8)\cong\langle\iota_s,\iota_v\rangle.O^+_8(2)~\to~ \left(\langle\iota_s,\iota_v\rangle.O^+_8(2)\right)/\langle \iota_v\rangle\cong W^+(E_8)
\end{equation}
is injective. Thus we have proved the following result. 
\begin{proposition}\label{prop:G0subgroupWplusE8}
For any four-torus sigma model the corresponding group $G_0$ is isomorphic to a subgroup of $W^+(E_8)$.
\end{proposition}

The discussion of this section is summarized in the following diagram.
\begin{equation}\label{diagram}
\begin{matrix}
\textrm{Spin}(8) &  \overset{\resizebox{6pt}{!}{$\pi$}}{\longrightarrow} & SO(8) \\
\hookuparrow & ~ & \hookuparrow   \\
~~~~\langle\iota_s\rangle.W^+(E_8)\cong\langle \iota_v,\iota_s\rangle.O_8^+(2)~~ & \longrightarrow & W^+(E_8)\cong\langle\iota_v\rangle.O_8^+(2) \\
\big\downarrow & ~ & \big\downarrow \pi'  \\
\langle\iota_s\rangle.O_8^+(2)\cong W^+(E_8) & \overset{\resizebox{11pt}{!}{$\pi''$}}{\longrightarrow} & O_8^+(2) \\
\hookuparrow & ~ & \hookuparrow   \\
~G_0 &\longrightarrow & ~G_1
\end{matrix}
\end{equation}

\section{The VOSA} 
\label{sec:e8model}
In this section we discuss the VOSA side of the VOSA/sigma model correspondence in this case: the $E_8$ lattice VOSA $\vfe$. In  \S\ref{subsec:The Theory} we introduce the theory and set up our notation, and in \S\ref{subsec:Twined Trace} we outline the computation of the twined traces of this VOSA, and prove the main theorm (Theorem \ref{thm:twining}) of the paper.

\subsection{The Theory}
\label{subsec:The Theory}
The VOSA $\vfe$ is a $c=12$ chiral superconformal field theory (SCFT) with eight free chiral fermions $\beta^a(z)$ and eight free chiral bosons $Y^a(z)$, with $a=1,\dots,8$. Moreover, it has chiral vertex operators $V_\lambda(z)=c(\lambda):e^{\lambda\cdot Y}:$ corresponding to the $E_8$ lattice. In the above, we have $\lambda\in E_8$ and $c(\lambda)$ is the standard operator needed for locality  \cite{FLM, DolanGoddardMontague}.
The stress tensor is given by 
\begin{equation}
T=-\sum_{a=1}^8:\partial Y^a~\partial Y^a:-\frac{1}{4}\sum_{a=1}^8 :\beta^a\partial\beta^{a}:, 
\end{equation}
and an ${\cal N}=1$ structure is provided by the 
supercurrent $Q$, proportional to the combination
\begin{equation}
\sum_{a=1}^8 :\beta^a\partial Y^a:~.
\end{equation}  

The $8$ currents $\partial Y^b$ form a $\mathfrak{u}(1)^8$ bosonic algebra, while the $28$ currents $:\beta^a\beta^b:$ generate a fermionic Kac-Moody algebra $\mathfrak{so}(8)_1$. 
Let $F$ be the eight-dimensional real vector space spanned by the fermions $\beta^a$. To facilitate the comparison with the sigma models, we split $F$ into two four-dimensional subspaces $F=X\oplus \bar X$ such that  $X$ is spanned by $\beta^a$ for $a=1,\dots, 4$ and $\bar X$ is spanned by $\beta^b$ for $b=5,\dots, 8$. 
As usual, it is convenient to work with the complex fermions 
\begin{gather}
\begin{split}
\gamma^i &:= \frac{1}{\sqrt{2}}(\beta^{i}+i\beta^{i+2})~, \\
\bar \gamma^i &:= \frac{1}{\sqrt{2}}(\beta^{i+4}+i\beta^{i+6})~,
\end{split}
\end{gather}
for $i=1,2$. 
The splitting of $F$ leads to the subalgebra $\mathfrak{so}(4)_1\oplus \mathfrak{so}(4)_1$ of  the fermionic Kac-Moody algebra $\mathfrak{so}(8)_1$.  
Focussing on the first $\mathfrak{so}(4)_1\cong\mathfrak{su}(2)_1\times \mathfrak{su}(2)_1$, corresponding to $X\subset F$, the two factors of $\mathfrak{su}(2)_1$ are generated by $J_X^{1,2,3}$ and $A_X^{1,2,3}$ respectively, completely analogous to the sigma model case 
(\eqref{R_symmetry_currents} and \eqref{def:sigma_su2A})
upon replacing the $\chi$s with $\gamma$s.

At the level of the zero-modes, 
we have 
\begin{equation}
\begin{split}
&SO(X)=(SU(2)^A_X\times SU(2)^J_{X})/(-1,-1)\cong SO(4)~, \\
&SO(\bar X)=(SU(2)^A_{\bar X}\times SU(2)^J_{\bar X})/(-1,-1)\cong SO(4)~.
\end{split}
\end{equation}
Note that all four $SU(2)$s above preserve the ${\cal N}=1$ superconformal algebra. 

Next we discuss the quantum states of the above model. 
We will sometimes refer to the space of states of this VOSA as an NS sector, since the chiral fermions satisfy the antiperiodic boundary condition. 
 One can also construct a canonically twisted module for this VOSA, i.e. a Ramond sector with periodic boundary conditions for the fermions. The Ramond sector contains $2^{8/2}=16$ ground states, forming a representation of the Clifford algebra of the fermionic zero modes. 
A convenient basis for these ground states may be denoted
 \begin{equation}
\label{twisted_gs}
|r\rangle:=|r_1,r_2,r_3,r_4\rangle~,~~r_1,r_2,r_3,r_4\in\left\lbrace\pm\frac{1}{2}\right\rbrace~. 
\end{equation}
Similar to the case of the sigma models \eqref{eqn:gndstates}, the Fock space ground states are then given by $|\lambda;r\rangle:=V_\lambda(0)|r\rangle$, where $\lambda\in E_8$.

With the sigma model elliptic genus \eqref{def:EG} in mind we define the following twisted module trace.
\be\label{def:VOSAtrace}
Z(\tau,z):=\Tr_{\textrm{tw}}\left[(-1)^{F}~y^{J_0^{X,3}}~q^{L_0-\frac{c}{24}}\right]\ 
\ee

The action of the operator ${J_0^{X,3}}$ on the oscillators and the ground states is completely analogous to its counterpart in the sigma models. Namely, it acts as a number operator for the fermionic oscillators, counting $\gamma^j_n$ excitations (with $n\leq-1$) as $+1$ and $\gamma^{j^*}_n$ excitations as $-1$, for $j=1,2$, while on the ground states \eqref{twisted_gs} it acts as
\begin{equation}
{J_0^{X,3}}|r\rangle=(r_1+r_2)|r\rangle~.
\end{equation}
Similarly, the fermion number operator is defined as $(-1)^F:=(-1)^{{J_0^{X,3}}+{J_0^{{\bar X},3}}}$, and acts on the ground states as
\begin{equation}
(-1)^F|r\rangle=(-1)^{{J_0^{X,3}}+{J_0^{{\bar X},3}}}|r\rangle=(-1)^{r_1+r_2+r_3+r_4}|r\rangle~.
\end{equation}
From this it follows immediately that states built on the ground states $|r\rangle$ with opposite signs of $r_3$ (or $r_4$) lead to opposite contributions to the trace $Z(\tau,z)$ and hence the trace vanishes. 
In the next subsection we will see that, similar to the sigma models, the trace is generically not vanishing when twined by a symmetry. 

\subsection{Twined Traces}
\label{subsec:Twined Trace}
Recall (Proposition \ref{prop:G0subgroupWplusE8}) that the symmetry groups $G_0$ of the four-torus sigma models may be regarded as subgroups of $W^+(E_8)$. 
We may thus identify them with symmetry groups of $\vfe$ which act on the $E_8$ lattice by even-determinant Weyl automorphisms, according to the vector representation $\rho_\psi$. 
The lattice $E_8$ is naturally contained in $F$, the 8-dimensional real vector space spanned by the fermions $\beta^a$, so we have $G_0< W^+(E_8)< SO(F)$. 
As discussed in \S\ref{subsec:sym_T4},  the groups $G_0$ are contained in an $SU(2)_L\times SU(2)_R$ subgroup of $SO(4)_L\times SO(4)_R\subset SO(8)$, and thus they do not mix the spaces $\Pi_L$ and $\Pi_R$. We can  further identify the vector spaces $X=\Pi_L$ and $\bar X=\Pi_R$, so that $G_0$ is contained in $SU(2)_X^A\times SU(2)_{\bar X}^A$ (and commutes with $SU(2)^J_X$ and $SU(2)^J_{\bar X}$) when acting on the $E_8$ lattice of the VOSA. The action of $G_0$ is then lifted to automorphisms of the $E_8$ VOSA that preserve the $\mathcal{N}=1$ supercurrent $Q$. 
(One may choose lifts where all phases are trivial. Consult \S\ref{app:cocycles} for details.)
As a result, for each $g\in G_0$ we may define the following $g$-twined trace in the twisted module for the $E_8$ VOSA
\be\label{eqn:VOSAtwinedPF}
Z_g(\tau,z):=\Tr_{\textrm{tw}}\left[g~(-1)^{F}~y^{J_0^{X,3}}~q^{L_0-\frac{c}{24}}\right]\ ,
\ee
generalising \eqref{def:VOSAtrace}. 

Analogous to the sigma models (\refeq{eigenvalues_vector}), the above   $g$-twined trace naturally decomposes into three factors,
\be\label{eqn:Zgthreefactors}
Z_g(\tau,z)=Z_g^{\text{osc}}(\tau,z)Z_g^{\text{gs}}(z)Z_g^{E_8}(\tau)~,
\ee
capturing the contribution from the oscillators, the fermionic ground states, and the $E_8$ lattice chiral operators, respectively. 

Choosing a convenient basis for the fermions we observe that the action of $g$ is precisely the same as in \eqref{def:gact_chi}, with $\chi^i$ replaced by $\gamma^i$ and $\chi^{i^*}$ replaced by $\gamma^{i^*}$, $\tilde\chi^i$ replaced by $\bar\gamma^{i}$ and $\tilde\chi^{i^*}$ replaced by $\bar\gamma^{i^*}$ for $i=1,2$.
As a result, the oscillators give a factor of
\small
\begin{equation}
\begin{split}
Z^{\text{osc}}_g(\tau,z)&=\prod_{n=1}^\infty \frac{(1-\zeta_L yq^n)(1-\zeta^{-1}_L yq^n)(1-\zeta_L y^{-1}q^n)(1-\zeta^{-1}_L y^{-1}q^n)(1-\zeta_R  q^n)^2(1-\zeta^{-1}_R  q^n)^2}{(1-\zeta_L q^n)^2(1-\zeta^{-1}_L q^n)^2(1-\zeta_R  q^n)^2(1-\zeta^{-1}_R  q^n)^2}\\
	&=\prod_{n=1}^\infty \frac{(1-\zeta_L yq^n)(1-\zeta^{-1}_L yq^n)(1-\zeta_L y^{-1}q^n)(1-\zeta^{-1}_L y^{-1}q^n)}{(1-\zeta_L q^n)^2(1-\zeta^{-1}_L q^n)^2}\ .
\end{split}
\end{equation}
\setlength{\parindent}{\baselineskip}
\setlength{\parskip}{0em}
\renewcommand{\baselinestretch}{1.3}

Similarly,  the group action on the fermionic ground states is given by
\begin{equation}
g|r\rangle=\zeta_L^{A_0^{X,3}}\zeta_R^{A_0^{\bar X,3}}|r\rangle=\zeta_L^{r_1-r_2}\zeta_R^{r_3-r_4}|r\rangle~,
\end{equation}
leading to the contribution
\be
Z^{\text{gs}}_g(\tau,z)=y^{-1}(1-\zeta_L y)(1-\zeta^{-1}_L y)(1-\zeta_R)(1-\zeta^{-1}_R)=2(1-\Re(\zeta_R))(y^{-1}+y-2\Re(\zeta_L))\ .
\ee

The contribution from the $E_8$ lattice is
\be
Z_g^{E_8}(\tau)=\sum_{\lambda\in \left(E_8\right)^{\rho_{\psi}(g)}} \xi_g(\lambda)~q^{\frac{\lambda^2}{2}}\ ,
\ee
where $\left(E_8\right)^{\rho_\psi(g)}$ is the sublattice of $E_8$ fixed by $g$ (which acts on the lattice according to the $\rho_\psi$ representation of $G_0$),
and $\xi_g(\lambda)$ are phases analogous to those in the sigma models \eqref{w-m_phases} that can be chosen to be trivial.

We now state and prove the main result of the paper.
\begin{theorem}
\label{thm:twining}
For every $g\in G_0$ for any of the possible groups $G_0$ we have
\be
Z_g(\tau,z)=\phi^G_g(\tau,z)\ .
\ee
\end{theorem}

\begin{proof}
To begin we note that, from the preceeding discussion, it is evident that for each $g\in G_0$ we have
\begin{equation}
Z_g^{\text{osc}}=\phi_g^{\text{osc}}~,~~Z_g^{\text{gs}}=\phi_g^{\text{gs}}~.
\end{equation}
So we require (see (\refeq{eigenvalues_vector}), (\ref{eqn:Zgthreefactors})) to show that $Z_g^{E_8}=\phi_g^{\text{w-m}}$.
Since  we have $Z_g^{\text{gs}}=\phi_g^{\text{gs}}=0$ whenever $\zeta_R=1$, we may focus solely on the case that $\zeta_R\neq1$. Moreover, if both $\zeta_L,\zeta_R\neq1$ then $Z_g^{E_8}=\phi_g^{\text{w-m}}=1$, as both lattices $\left(E_8\right)^{\rho_\psi(g)}$ and $\left(\Gamma^{4,4}_\textrm{w--m}\right)^{g}$ are empty in this case. Therefore, we only need to prove that whenever $\zeta_L=1$ and $\zeta_R\neq 1$, the fixed sublattice $\left(E_8\right)^{\rho_\psi(g)}$ is isomorphic to $\left(\Gamma^{4,4}_\textrm{w--m}\right)^{g}$. We will achieve this by performing a case-by-case analysis. There are only four classes in $\rho_\psi$ with $\zeta_L=1$ and $\zeta_R\neq 1$. In the notation explained in \S\ref{app:sym}, they are  2A, 2E, 3E, 4A (see Table \ref{t:classes}).

To proceed we note that by inspecting the character table of $W^+(E_8)$ we may deduce that the aforementioned classes are necessarily fixed by the action of any outer automorphism. 
Since the representations $\rho_\psi$ and $\rho_{e}$ are related by such triality outer automorphisms (cf. \S\ref{sec:symmetry groups}), we deduce that for these classes we have $\left(E_8\right)^{\rho_\psi(g)}\cong \left(E_8\right)^{\rho_e(g)}$, the latter being the lattice fixed by $\langle g \rangle \subseteq G_0$ in the representation $\rho_e$. In \S4 of \cite{t4paper}, both lattices $\left(E_8\right)^{\rho_e(g)}$ and  $\left(\Gamma^{4,4}_\textrm{w--m}\right)^{g}$ were described in detail. In particular, it was shown that they are 
as in (\ref{eqn:rankfourlattices}).
\begin{equation}\label{eqn:rankfourlattices}
\begin{array}{c|c|c|c|c}
~ & 2A & 2E & 3E & 4A \\
\hline \left(E_8\right)^{\rho_e(g)}\left(\cong \left(E_8\right)^{\rho_\psi(g)}\right) & D_4 & A_1^4 & A_2^2 & D_4 \\
\hline \left(\Gamma^{4,4}_\text{w--m}\right)^{g} & D_4 & A_1^4 & A_2^2 & D_4
\end{array}~~
\end{equation}
From (\ref{eqn:rankfourlattices}) we see that the fixed sublattice of the winding-momentum lattice of the four-torus sigma model and the fixed sublattice of the $E_8$ lattice are isomorphic in each case. This completes the proof.
\end{proof}

\section{Orbifolds} 
\label{sec:orbifolds}

In this section we investigate the extent to which the diagram Figure \ref{fig:Z2diagram} commutes, or not, with an arbitrary symmetry in place of the specific $\ZZ_2$ action indicated. We will demonstrate that in fact the diagram commutes for all possible choices, at least if we assume a certain claim about orbifolds of four-torus sigma models.
We regard this result---Proposition \ref{pro:orbifold}---as further evidence that the VOSA/sigma model correspondence for four-torus sigma models proposed herein represents a natural structure. 

The claim about orbifold sigma models we will require to assume is the statement that: 
\begin{quote}
{\it The orbifold of a four-torus sigma model 
by a discrete supersymmetry preserving symmetry is either a sigma model with $T^4$ target or a sigma model with K3 target.}
\end{quote} 
This claim follows, for example, from the conjecture 
that
the only  ${\cal N}=(4,4)$ SCFTs with four spectral flow generators, central charge $c=\bar c=6$ and discrete spectrum come from sigma models with $T^4$ or K3 target space. This conjecture is widely believed to be true (see e.g. \cite{hiker}) and was implicitly assumed in early string theory literature. Here we refer to it as the {\em uniqueness conjecture}.

Alternatively, the above claim on four-torus sigma model orbifolds is supported by the following heuristic argument which is independent of the uniqueness conjecture.
Call a  
symmetry $g$ of a sigma model $\cT$ with target $X$ {\em geometric} 
if it is lifted (cf. \S\ref{app:cocycles}) from a symmetry $\bar g$ of the target space $X$. Then the 
orbifold 
of $\cT$ by $g$ should be a sigma model on the orbifold 
of $X$ by $\bar g$. Any orbifold of a four-torus is a singular limit of K3 surfaces, so the claim about orbifolds should hold at least for geometric symmetries. 

For more general symmetries note that it can be shown, independently of the uniqueness conjecture (see e.g. \cite{hiker}), that the elliptic genus of an ${\cal N}=(4,4)$ SCFT 
with four spectral flow generators
and $c=\bar c=6$ 
is either $0$ or coincides with the K3 elliptic genus. 
Furthermore, if the elliptic genus is 0 then the corresponding sigma model has $T^4$ target\cite{hiker}. 
So, if the elliptic genus of an orbifold is $0$, there is no doubt that it is a sigma model on $T^4$. 

To handle the case that the elliptic genus of the orbifold is non-vanishing we recall the {\em reverse orbifold construction}: 
If $\cT$ is a sigma model and $g$ is a discrete supersymmetry preserving symmetry of $\cT$ then the orbifold $\cT'$ of $\cT$ by $g$ 
has a distinguished symmetry $g'$ with the property that the orbifold of $\cT'$ by $g'$ is $\cT$. 
(See e.g. \cite{lamshi} for an analysis of this
in the VOA setting.)

The supersymmetry preserving symmetries of sigma models with K3 target have been classified in \cite{gaberdiel_k3}, and this allows us to determine the pairs $(\cT',g')$, with $\cT'$ a K3 sigma model and $g'$ a symmetry of $\cT'$, for which the orbifold of $\cT'$ by $g'$ is a sigma model on $T^4$. (One just checks if the elliptic genus of the orbifold vanishes or not.)
So by the reverse orbifold construction we obtain a corresponding set of pairs $(\cT,g)$, with $\cT$ a sigma model on $T^4$ and $g$ a symmetry of $\cT$, for which $\cT$ is an orbifold of a K3 sigma model $\cT'$ and $g$ is the corresponding distinguished symmetry such that the orbifold of $\cT$ by $g$ is $\cT'$. Finally, we can check case-by-case that every non-geometric four-torus sigma model symmetry, for which the corresponding orbifold elliptic genus is non-vanishing, occurs in such a pair. 
So there are simply no candidates for four-torus sigma model orbifolds by non-geometric symmetries with non-vanishing elliptic genus except for K3 sigma models.

Note that the claim above on four-torus sigma model orbifolds has a rigorous counterpart for VOSAs.
Namely, if $\hat g\in\Aut(\vfe)$ is the standard lift (cf. \S\ref{app:cocycles}) of a four-torus sigma model symmetry $g\in W^+(E_8)$ 
then the orbifold of $\vfe$ by $\hat g$ is either isomorphic to $\vfe$ or to $V^{s\natural}$, the latter being the unique ${\cal N}=1$ VOSA with $c=12$ and vanishing weight $\frac12$ subspace \cite{Duncan07,conway}. 
We will establish this in the course of proving our next result, Proposition \ref{pro:orbifold}.
Note that a more general orbifolding of $\vfe$ might result in the VOSA that describes $24$ free fermions. Cf. e.g. \cite{refl}.

We now prove the main result of this section. For the formulation of this we assume the notation of (\ref{G_symmetries}).

\begin{proposition}
\label{pro:orbifold}
Let $\cT$ be a four-torus sigma model and let $g\in G_0<W^+(E_8)$ be a symmetry of $\cT$ that preserves the ${\cal N}=4$ superconformal algebra. 
Let $\hat g$ denote the standard lift of $g<W^+(E_8)$ to a symmetry of the VOSA $\vfe$ as described in \S\ref{app:cocycles}. 
If we assume that any orbifold of a four-torus sigma model by a discrete supersymmetry preserving symmetry is either a sigma model on $T^4$ or a sigma model on $K3$ then the orbifold of $\vfe$ by $\hat g$ is isomorphic to $\vfe$ or $\vsn$ according as the orbifold of $\cT$ by $g$ is a sigma model on $T^4$ or a sigma model on $K3$.
\end{proposition}

\begin{proof}
The orbifold of $\vfe$ by $\hat g$ is either $\vfe$ or $\vsn$ or the VOSA associated to $24$ free fermions according to Theorem 3.1 of \cite{refl}. To 
tell the three possibilities apart we can simply compute 
the partition function $Z_{\hat g\text{-orb}}(\tau)$ of the orbifold theory. It will develop that either $Z_{\hat g\text{-orb}}(\tau) = Z(\vfe;\tau)$ or $Z_{\hat g\text{-orb}}(\tau) = Z(V^{s\natural};\tau)$, where $Z(V;\tau)$ is the partition 
function of $V$. (In particular, the free fermion model will not arise.)

Let us denote the anti-periodic and periodic boundary conditions for the fermions by A and P, respectively. We are interested in the case where the fermions are in the $[A,A]$ sector.  The bosons will always have periodic boundary condition in the current context so we will not explicitly specify the boson boundary condition in what follows. 

Let $_D^{\tilde D} \hspace{-0.5pt}Z_h^g(\tau)$ denote the $h$-twisted, $g$-twined partition function of $\vfe$ in the sector where the fermions have $[D,\tilde D]$ boundary conditions, with $D,\tilde D\in\lbrace A,P\rbrace$. The orbifold partition function is then given by
\begin{equation}
\label{orb_part}
Z_{\hat g\text{-orb}}(\tau)=\frac{1}{|\hat g|}~\sum_{k,\ell\in \ZZ/|\hat g|}^{}     ~ _A^A\hspace{-0.5pt} Z_{\hat g^k}^{\hat g^\ell} (\tau),
\end{equation}
so we need to compute $_D^{\tilde D} \hspace{-0.5pt}Z_{\hat g^k}^{\hat g^\ell}$ for all $k,\ell\in \ZZ/|\hat g|$. (We have $|\hat{g}|=|g|$ for all cases except for when $g$ lies in the class $2E$, in which case $|\widehat{2E}|=2|2E|=4$. More details on this can be found in \S\ref{app:cocycles}.)

Recall that modular transformations changes  the twisting and twining boundary conditions according to 
\begin{equation}
\label{Z_modular}
\text{PSL}_2(\mathbb{Z}) \ni \gamma=\left(\begin{matrix} a & b \\ c & d \end{matrix}\right) : (h,g)\mapsto (g^c h^d,g^ah^b)
\end{equation}
Notice that $\gamma\in\text{PSL}_2(\mathbb{Z})$ implies that $(h,g)$ and $(h^{-1},g^{-1})$ correspond to equal partition functions, since in our case all fields are invariant (self-conjugate) under charge conjugation $C=S^2=(ST)^3$.
Additionally, modular transformations also mix the fermionic sectors $[A,A], [A,P], [P,A]$, while leaving the bosonic sector $[P,P]$ invariant. In particular, for a holomorphic VOSA of central charge $c$, the partition functions $_A^AZ$,  $_A^PZ$,  $_P^AZ$ span a $3$-dimensional representation $\rho_c:PSL_2(\ZZ)\to GL(3)$ given by
\begin{equation}
\label{ferm_mix}
\begin{split}
&\left(\begin{matrix}  _A^AZ\\  _A^PZ \\  _P^AZ \end{matrix}\right)\left(-\frac{1}{\tau}\right) = \rho_c(S) \left(\begin{matrix}  _A^AZ \\  _A^PZ\\  _P^AZ\end{matrix}\right)\left({\tau}\right) ~,\qquad \rho_c(S)=\left(\begin{matrix} 1 & 0 & 0  \\ 0 & 0 & 1  \\ 0 & 1 & 0  \end{matrix} \right)~, \\
&\left(\begin{matrix}  _A^AZ\\  _A^PZ \\  _P^AZ\end{matrix}\right)\left(\tau+1\right) = \rho_c(T) \left(\begin{matrix}  _A^AZ \\  _A^PZ \\  _P^AZ \end{matrix}\right) \left(\tau\right) ~,\qquad\rho_c(T)=\left(\begin{matrix} 0 & \text{e}(-{c\over 24}) & 0  \\ \text{e}(-{c\over 24}) & 0 & 0  \\ 0 & 0 & \text{e}({c\over 12}) \end{matrix} \right)~.
\end{split}
\end{equation}
Combining the above, we conclude that 
\be\label{eqn:twist_twine_transform}
\left(\begin{matrix}  _A^A\hspace{-0.5pt}Z_{\hat g^m}^{\hat g^n}\\  _A^P\hspace{-0.5pt}Z_{\hat g^m}^{\hat g^n}\\  _P^A\hspace{-0.5pt}Z_{\hat g^m}^{\hat g^n} \end{matrix}\right)(\tau)  = \ex(\alpha) ~\rho^{-1}_c\!\left(\left(\begin{smallmatrix}a&b\\c&d\end{smallmatrix}\right)\right)~\left(\begin{matrix}  _A^A\hspace{-0.5pt}Z_{1}^{g'}\\  _A^P\hspace{-0.5pt}Z_{1}^{g'} \\  _P^A\hspace{-0.5pt}Z_{1}^{g'} \end{matrix}\right)\!\left({a\tau+b\over c\tau+d}\right),
\ee
for some $(\begin{smallmatrix}a&b\\c&d\end{smallmatrix})\in\text{PSL}_2(\mathbb{Z})$ that can be determined from (\ref{Z_modular}), some
$g'\in \langle \hat g \rangle$ and some phase $\ex(\alpha):=e^{2\pi i \alpha}$. 

Let us use the fact that the VOSA $\vfe$ is the product of a (bosonic) holomorphic lattice VOA based on the $E_8$ lattice, and the VOSA generated by $8$ real (or four complex) free fermions, and that the symmetry $\hat g$ acts independently on these two algebras. As a consequence, the twisted-twined partition functions $ _A^A\hspace{-0.5pt}Z_{\hat g^k}^{\hat g^\ell}$ factorize as
\begin{equation}
_A^A\hspace{-0.5pt}Z_{\hat g^k}^{\hat g^\ell}= ~_A^A\hspace{-0.5pt}F_{\hat g^k}^{\hat g^\ell}B_{\hat g^k}^{\hat g^\ell}
\end{equation} into the product of the twisted-twined partition functions $~_A^A\hspace{-0.5pt}F_{\hat g^k}^{\hat g^\ell}$ and $B_{\hat g^k}^{\hat g^\ell}$ of the fermionic VOSA (with $[A,A]$ boundary conditions)  and the bosonic VOA, respectively.

We will consider the fermion and boson contributions separately, and then combine the results. 
Consider first the four free complex fermions, with $c_F=4$.  Let us denote the partition function in sector $[D,\widetilde{D}]$ by $^{\widetilde{D}}_DF$. Then we have
\begin{equation}
\label{free_ferm_part}
\begin{split}
&_A^AF(\tau)=
{\theta_3^4(\tau) \over \eta^4(\tau)}~, 
\\
&_A^PF(\tau)=
{\theta_4^4(\tau) \over \eta^4(\tau)}~, 
\\
&_P^AF(\tau)=
{\theta_2^4(\tau) \over \eta^4(\tau)}~, 
\\
&_P^PF(\tau)=
{\theta_1^4(\tau) \over \eta^4(\tau)} =0~,
\end{split}
\end{equation}
where we write $\theta_i(\tau,z)$ for the usual Jacobi theta functions and set $\theta_i(\tau):=\theta_i(\tau,0)$.
The sectors $_A^AF(\tau),_A^PF(\tau)_P^AF(\tau)$ transform as in \eqref{ferm_mix} under $PSL_2(\ZZ)$, with $c=4$.

Now consider a symmetry $\hat g$ acting on the fermions, with eigenvalues determined by the representation $\rho_\psi$, and denoted $\zeta_L=\ex(\alpha_L)$ and $\zeta_R=\ex(\alpha_R)$, where $\zeta_L$ and $\zeta_R$ are as in (\ref{def:gact_chi}). 
 Then the $\hat g^k$-twisted $\hat g^\ell$-twined partition function in the four sectors is given by
\begin{equation}\label{ferm_part}\begin{split}
_A^AF_{\hat g^k}^{\hat g^\ell}(\tau)
&=q^{(\hat\alpha^2_L+\hat\alpha_R^2)k^2}{\theta_3^2(\tau,\hat\alpha_L(k\tau+\ell)) \,\theta_3^2(\tau,\hat\alpha_R(k\tau+\ell)) \over \eta^4(\tau)}~,\\
_A^PF_{\hat g^k}^{\hat g^\ell}(\tau)
&=q^{(\hat\alpha^2_L+\hat\alpha_R^2)k^2}{\theta_4^2(\tau,\hat\alpha_L(k\tau+\ell)) \,\theta_4^2(\tau,\hat\alpha_R(k\tau+\ell)) \over \eta^4(\tau)}~,\\
_P^AF_{\hat g^k}^{\hat g^\ell}(\tau)
&=q^{(\hat\alpha^2_L+\hat\alpha_R^2)k^2}{\theta_2^2(\tau,\hat\alpha_L(k\tau+\ell)) \,\theta_2^2(\tau,\hat\alpha_R(k\tau+\ell)) \over \eta^4(\tau)}~,\\
_P^PF_{\hat g^k}^{\hat g^\ell}(\tau)
&=q^{(\hat\alpha^2_L+\hat\alpha_R^2)k^2}{\theta_1^2(\tau,\hat\alpha_L(k\tau+\ell)) \,\theta_1^2(\tau,\hat\alpha_R(k\tau+\ell)) \over \eta^4(\tau)}~,
\end{split}\end{equation}
where $0\le k,\ell<N$, and $\hat\alpha_{L,R}\equiv \alpha_{L,R}(k)$ are rational numbers such that $e(\hat\alpha_{L,R})=\zeta_{L,R}$ and $-\frac{1}{2}< \hat\alpha_Lk,\hat\alpha_Rk\le \frac{1}{2}$.  Up to a possible redefinition $\zeta_L\leftrightarrow \zeta_L^{-1}$ or $\zeta_R\leftrightarrow \zeta_R^{-1}$, one can restrict $0\le \hat\alpha_Lk,\hat\alpha_Rk\le \frac{1}{2}$. Notice that the expressions \eqref{ferm_part} are in general not invariant under $k\to k+N$ and $\ell\to \ell+N$, but they can change by a multiplicative constant phase (an $N$-th root of unity). This phenomenon reflects an ambiguity in the definition of the phases of $_D^DF_{\hat g^k}^{\hat g^\ell}$, that depend on the choice  of the action of $\langle \hat g\rangle$ on the $\hat g^k$-twisted module.

Next we consider  four free complex bosons on the $E_8$ torus, with $c_B=8$. The bosons naturally have periodic boundary conditions on both cycles of the torus. The corresponding partition function is
\begin{equation}
B(\tau):=\frac{\Theta_{E_8}(\tau)}{\eta(\tau)^8}~,
\end{equation}
where 
\be\Theta_{E_8}(\tau)=\frac{1}{2}\left(\theta_2(\tau)^8+\theta_3(\tau)^8+\theta_4(\tau)^8\right)=E_4(\tau)\label{eqn:thetaseriese8}
\ee
 is the theta series of the $E_8$ lattice, equal to the Eisenstein series of weight $4$. Under modular transformations the partition function transforms according to
\begin{equation}
\label{free_bosons_modular}
B(-\tfrac{1}{\tau})=B(\tau)~,~~B(\tau+1)=\text{e}(-\tfrac{1}{3})B(\tau)~.
\end{equation}

A symmetry $\hat g$ acts on the four complex bosons in the same way as for the fermions, leaving invariant the supersymmetry of the $E_8$ VOSA. 
The corresponding untwisted $\hat g^n $-twined partition function is thus given by
\small
\begin{equation}\label{eqn:twined_boson}
B_1^{\hat g^n}(\tau)=q^{-\frac{1}{3}}\prod_{k=1}^{\infty}\left(\frac{1}{1-\zeta_L^n  q^k}\right)^2\left(\frac{1}{1-\zeta_L^{-n} q^k}\right)^2\left(\frac{1}{1-\zeta_R^n  q^k}\right)^2\left(\frac{1}{1-\zeta_R^{-n} q^k}\right)^2 \Theta_{\Lambda^{\hat g^n}}(\tau)~,
\end{equation}
where $\Theta_{\Lambda^{\hat g^n}}(\tau)$ is the theta series of the sublattice fixed by $\hat g^n$ (except for the case that $g$ is of class $2E$, and $n=2$, wherein $\Theta_{\Lambda^{\hat g^n}}$ takes a slightly different meaning, as explained below).
When $\zeta_L^n , \zeta_R^n\neq 1$, one has $\Theta_{\Lambda^{\hat g^n}}=1$ and the above may be conveniently written as
\be B_1^{\hat g^n}(\tau)= 
(\zeta_L^{\frac{n}{2}}-\zeta_L^{-\frac{n}{2}})^2(\zeta_R^{\frac{n}{2}}-\zeta_R^{-\frac{n}{2}})^2 \frac{\eta(\tau)^4}{\theta_1^2(\tau,n\alpha_L)\theta_1^2(\tau,n\alpha_R)}.
\ee
The cases for which $\Theta_{\LL^{\hat g^n}}(\tau)$ is not identically $1$ are summarized in Table \ref{fixed_sublattices}, so that $\Theta_{\LL^{\hat g^n}}$ is the theta series of the $D_4$ lattice, for example, when $g$ is of class $2A$ or $4A$ and $n=1$. As hinted above, the case that $g$ belongs to $2E$ and $n=2$ is a bit more subtle. This is because $\hat g^2$ is non-trivial, even though $g$ has order $2$. We have
\begin{equation}
\hat{g}^2 \left(V_\lambda\right)=(-1)^{(\lambda,g(\lambda))} V_\lambda~,
\end{equation}
and the result of this is that $\Theta_{\LL^{\hat g^2}}$ should be interpreted as $\Theta_{\widetilde{E}_8}(\tau):=\theta_3^4(\tau)\theta_4^4(\tau)$, rather than just the theta series (\ref{eqn:thetaseriese8}) of $E_8$, when $g$ is of class $2E$.
\begin{table}[htbp]
	\small
	\centering
	{\renewcommand{\arraystretch}{1.5}
		\begin{tabular}{c|c|c|c|c|c|c|c}
			~ &  $\widehat{2A}$ & $\widehat{2E}$ & $\widehat{3E}$ &  $\widehat{4A}$ & $\widehat{-4A}$ & $\widehat{-3E}$ & $\widehat{6BC}$ \\ \hline\hline
			$\hat{g}$ &  $D_4$ & $A_1^4$  & $A_2^2$  & $D_4$ & --- & --- & --- \\ \hline
			$\hat{g}^2$ & ~ & $\widetilde{E_8}$  & $A_2^2$ & $D_4$ & $D_4$ & $A_2^2$ & --- \\ \hline
			$\hat{g}^3$ & ~ & ~  & ~ & ~ & ~ & --- & $D_4$ \\ 
	\end{tabular}}
	\caption{\small Fixed sublattices of $E_8$ in $\rho_\psi$, by powers of conjugacy classes of $W^+(E_8)$}
	\label{fixed_sublattices}
\end{table}
\setlength{\parindent}{\baselineskip}
\setlength{\parskip}{0em}
\renewcommand{\baselinestretch}{1.3}

The whole set of bosonic twisted-twined partition functions $B_{\hat g^k}^{\hat g^\ell}$ can be recovered from the untwisted ones $B_{1}^{\hat g^n}$ using the analog  of \eqref{eqn:twist_twine_transform} for the bosonic case, namely
\be\label{eqn:twist_twine_bosonic}
B_{\hat g^m}^{\hat g^n}(\tau) = \ex(\alpha_B) B_{1}^{g'}\!\left({a\tau+b\over c\tau+d}\right)\ .
\ee
for some phases $\ex(\alpha_B):=e^{2\pi i \alpha_B}$.

We need to have some control over the phases $\ex(\alpha_B)$ in (\ref{eqn:twist_twine_bosonic}). For orbifolds of holomorphic VOAs by cyclic groups, these phases were discussed in \cite{holomvoas}. More precisely, if $V$ is a simple, rational, $C_2$-cofinite, self-contragredient vertex operator algebra and $g$ is an automorphism of $V$ of order $N$ then the phases are governed by a 2-cocycle representing a class in $H^2(\ZZ_N,\ZZ_N)\cong \ZZ_N$. According to Proposition 5.10 of \cite{holomvoas}, the cohomology class depends on $2N^2\rho_1 \mod N$, where $\rho_1$ is the conformal weight of the irreducible $g$-twisted $V$-modules $V(g)$. Different cocycles in the same class correspond to different choices for the action of $\langle g\rangle$ on the twisted sectors.

It turns out that, upon combining the fermions and bosons into the full twisted twined partition functions $_D^D\hspace{-0.5pt}Z_{\hat g^k}^{\hat g^\ell}= ~_D^D\hspace{-0.5pt}F_{\hat g^k}^{\hat g^\ell}B_{\hat g^k}^{\hat g^\ell}$,  the phases $\ex(\alpha_B)$ always cancel against the analogous phases for the fermionic contribution, so that the phases $\ex(\alpha)$ in \eqref{eqn:twist_twine_transform} are trivial.

For example, when  $\zeta_L^{n},\zeta_R^{n}\neq 1$, where $n=\gcd(k,\ell)$, one obtains
\begin{equation}
B_{\hat g^k}^{\hat g^\ell}(\tau)=(\zeta_L^{\frac{n}{2}}-\zeta_L^{-\frac{n}{2}})^2(\zeta_R^{\frac{n}{2}}-\zeta_R^{-\frac{n}{2}})^2q^{-(\hat\alpha_L^2+\hat\alpha_R^2)k^2}\frac{\eta(\tau)^4}{\theta_1^2(\tau,\hat\alpha_L(k\tau+\ell))\theta_1^2(\tau,\hat\alpha_R(k\tau+\ell))}
\end{equation} where $0\le \hat \alpha_L,\hat\alpha_R\le 1/2$, so that, combining the fermions and bosons, we obtain
\begin{equation}\label{complete}\begin{split}
_A^A\hspace{-0.5pt}Z_{\hat g^k}^{\hat g^\ell}&=(\zeta_L^{\frac{n}{2}}-\zeta_L^{-\frac{n}{2}})^2(\zeta_R^{\frac{n}{2}}-\zeta_R^{-\frac{n}{2}})^2\frac{\theta_3^2(\tau,\alpha_L(k\tau+\ell))\theta_3^2(\tau,\alpha_R(k\tau+\ell))}{\theta_1^2(\tau,\alpha_L(k\tau+\ell))\theta_1^2(\tau,\alpha_R(k\tau+\ell))}\ ,\\
_A^P\hspace{-0.5pt}Z_{\hat g^k}^{\hat g^\ell}&=(\zeta_L^{\frac{n}{2}}-\zeta_L^{-\frac{n}{2}})^2(\zeta_R^{\frac{n}{2}}-\zeta_R^{-\frac{n}{2}})^2\frac{\theta_4^2(\tau,\alpha_L(k\tau+\ell))\theta_4^2(\tau,\alpha_R(k\tau+\ell))}{\theta_1^2(\tau,\alpha_L(k\tau+\ell))\theta_1^2(\tau,\alpha_R(k\tau+\ell))}\ ,\\
_P^A\hspace{-0.5pt}Z_{\hat g^k}^{\hat g^\ell}&=(\zeta_L^{\frac{n}{2}}-\zeta_L^{-\frac{n}{2}})^2(\zeta_R^{\frac{n}{2}}-\zeta_R^{-\frac{n}{2}})^2\frac{\theta_2^2(\tau,\alpha_L(k\tau+\ell))\theta_2^2(\tau,\alpha_R(k\tau+\ell))}{\theta_1^2(\tau,\alpha_L(k\tau+\ell))\theta_1^2(\tau,\alpha_R(k\tau+\ell))}\ .
\end{split}
\end{equation} Using the modular properties of Jacobi theta functions, it is easy to verify that \eqref{eqn:twist_twine_transform} holds with $\rho_c$ given by \eqref{ferm_mix} with $c=12$  and with trivial phases $\ex(\alpha)$. An analogous result holds when $\zeta_L^{n}=1$ or $\zeta_R^{n}= 1$, $n=\gcd(k,\ell)$, although the formulae \eqref{complete} are not valid in this case.

Combining the above we may verify case-by-case that $Z_{\hat g\text{-orb}}(\tau) = Z(\vfe;\tau)$ whenever the $g$-orbifold of the four-torus sigma model is again a four-torus sigma model, and $Z_{\hat g\text{-orb}}(\tau) = Z(V^{s\natural};\tau)$ whenever the $g$-orbifold of the four-torus sigma model is a K3 sigma model, which is what we required to show.
 \end{proof}

\section{Reflection}
\label{sec:reflection}

The procedure of reflection on a non-chiral theory entails mapping all right-movers to left-movers, resulting in a holomorphic theory that may or may not be consistent. In \cite{k3refl} such a procedure was used to show that the K3 sigma model with $\ZZ_2^8:\mathbb{M}_{20}$ symmetry can be consistently reflected to give the Conway moonshine module VOSA $V^{s\natural}$. Moreover, the necessary and sufficient conditions that allow for reflection in a general theory were studied in detail. 

In this section we demonstrate that a similar reflection relation holds between a specific four-torus sigma model and the VOSA $\vfe$. In other words, we verify that Property \ref{property-ref} 
of VOSA/sigma model correspondences
holds for $\vfe$ and four-torus sigma models. 
To formulate this result precisely we first note that, according to \cite{t4paper}, there exists a unique point $\mu^\ast\in {\cal M}(T^4)$ such that the corresponding sigma model $\Sigma(T^4;\mu^\ast)$ has $G_0 \cong \mathcal{T}_{24}\times_{C_3}\mathcal{T}_{24}$. Now we may state the main result of this section.
\begin{theorem}
\label{thm:reflection}
The image of $\Sigma(T^4;\mu^\ast)$ under the reflection operation is a VOSA isomorphic to $\vfe$.
\end{theorem}

For the proof of Theorem \ref{thm:reflection} it will be convenient to use a quaternionic description of  the relevant lattices.
Let $\mathbb{H}$ be the space of quaternions, and write $\mathbf{i},\mathbf{j},\mathbf{k}$ for the imaginary units satisfying the usual quaternionic multiplication rule. Then $q\in \mathbb{H}$ can be written as $q=q_1+q_2\mathbf{i}+q_3\mathbf{j}+q_4\mathbf{k} $, where $q_1,q_2,q_3,q_4\in \RR$. We will often denote an element $q\in\mathbb{H}$ in terms of its components $(q_1,q_2,q_3,q_4)\in \RR^4$, and write $q=(q_1,q_2,q_3,q_4)$. We use the following norm on $\mathbb{H}$: 
\begin{equation}
||q||^2:=\sum_{i=1}^4q_i^2~,
\end{equation}
and the following notation for elements of $\mathbb{H}^2$ and $\mathbb{H}^{1,1}$
\begin{equation}
\begin{split}
& \mathbb{H}^2\ni (p|q):= (p_1,p_2,p_3,p_4|q_1,q_2,q_3,q_4)~, \\
& \mathbb{H}^{1,1}\ni (p;q):= (p_1,p_2,p_3,p_4;q_1,q_2,q_3,q_4)~, 
\end{split}
\end{equation}
where the corresponding norms are given by
\begin{equation}
||(p|q)||^2:=\sum_{i=1}^4 p_i^2+q_i^2~,~~ ||(p;q)||^2:=\sum_{i=1}^4 p_i^2-q_i^2~.
\end{equation}

The following lemma details a quaternionic realisation of the $E_8$ lattice.
\begin{lemma}\label{lattice_lem}
The eight-dimensional lattice defined by
\begin{equation}
\label{E8_wm}
\Gamma^8_{\text{w-m}}=\left\lbrace \frac{1}{\sqrt{2}}(a|b)~|~a_i,b_i\in\ZZ,~\sum_{i=1}^4 b_i\in 2\ZZ,~ a_i-b_i\equiv a_j-b_j\mod2~~\forall~~i,j\in\lbrace1,2,3,4\rbrace \right\rbrace~
\end{equation}
is a copy of the $E_8$ lattice. 
\end{lemma}

\begin{proof}
Recall that the Hurwitz quaternions are defined by 
\begin{equation}  
\mathcal{H}=\left\lbrace q\in\mathbb{H}~|~ (q_1,q_2,q_3,q_4)\in\ZZ^4\cup\left(\ZZ+\frac{1}{2}\right)^4 \right\rbrace\subset\mathbb{H}~.
\end{equation}
Then, according to \S2.6 of \cite{packings}, for example, we obtain a copy of the $E_8$ lattice in $\mathbb{H}^2$ by considering
\begin{equation}
\label{E8_hurw}
\Lambda_{E_8}\cong\left\lbrace \frac{p}{\sqrt{2}}(2|0)+\frac{q}{
\sqrt{2}}(1-\mathbf{i}|1-\mathbf{i})~~|~~p,q\in\mathcal{H} \right\rbrace, 
\end{equation}
where we write $q'(p|q) := (q'p | q'q)$.
In this realisation the $240$ roots of $E_8$ are expressed as follows,
\begin{equation}
\label{e8_roots}
\begin{split}
16 \text{~~roots of the form~~~} & \frac{1}{\sqrt{2}}(\pm2,0,0,0|0,0,0,0)~, \\
32 \text{~~roots of the form~~~} & \frac{1}{\sqrt{2}}(\pm1,\pm1,\pm1,\pm1|0,0,0,0)~, \\
192 \text{~~roots of the form~~~} & \frac{1}{\sqrt{2}}(\pm1,\pm1,0,0|\pm1,\pm1,0,0)~,
\end{split}
\end{equation}
where at the first line the $\pm2$ can be in any position, at the second line the four factors of $\pm1$ can be either all at the left or all at the right, and at the last line the pair of $\pm1$ at the right can either be at the same positions as the pair at the left or at complementary positions.

We claim that the sets defined by (\refeq{E8_hurw}) and (\refeq{E8_wm}) are the same. For this note that in terms of components we have
\be
p(2|0) + q (1-\mathbf{i}|1-\mathbf{i}) = 2(p_1,p_2,p_3,p_4|0)+ (q_1+q_2,-q_1+q_2, q_3-q_4,  q_3+q_4|q_1+q_2,-q_1+q_2, q_3-q_4,  q_3+q_4 ),
\ee
and it follows that $\Lambda_{E_8}\subseteq \Gamma^8_{\text{w-m}}$. To check that $ \Gamma^8_{\text{w-m}}\subseteq \Lambda_{E_8}$, we define, for every $ \frac{1}{\sqrt{2}}(a|b) \in  \Gamma^8_{\text{w-m}}$, 
\be
p_i := a_i -b_i , ~i\in \{1,2,3,4\}
\ee
and 
\be
q_{2i-1}:=  {1\over 2}(b_{2i-1}-b_{2i}), ~q_{2i}:=  {1\over 2}(b_{2i-1}+b_{2i}), ~~i\in \{1,2\}.
\ee
Then the condition $a_i-b_i\equiv a_j-b_j\mod2$ guarantees that $(p_1,p_2,p_3,p_4)\in\ZZ^4\cup\left(\ZZ+\frac{1}{2}\right)^4$, and  the condition $\sum_{i=1}^4 b_i\in 2\ZZ$ guarantees that  $(q_1,q_2,q_3,q_4)\in\ZZ^4\cup\left(\ZZ+\frac{1}{2}\right)^4$. This finishes the proof. 
\end{proof}

Now we are ready to prove Theorem \ref{thm:reflection}. 
\begin{proof}[Proof of Theorem \ref{thm:reflection}.]
 Recall that $\Sigma(T^4;\mu)$ has a simple description in terms of Fock space oscillators and vertex operators based on the winding-momentum lattice $\Gamma_{\rm w-m}(\mu)$ corresponding to the point $\mu \in {\cal M}(T^4)$. Since all right-moving oscillators are straightforwardly reflected to left-moving ones, the only non-trivial part of the proof is to show that the reflection of the winding-momentum lattice $\Gamma_{\rm w-m}(\mu^\ast)$ is isomorphic to the $E_8$ lattice.

At the moduli point $\mu^\ast$ of four-torus sigma model labelled by $\Lambda_{D_4}$, where the symmetry group is given by $G_0=\mathcal{T}_{24}\times_{C_3}\mathcal{T}_{24}$ in the notation of \cite{t4paper}, the even unimodular winding-momentum lattice is given in quaternionic language by
\begin{equation}
\Gamma^{4,4}_{\text{w-m}}=\left\lbrace \frac{1}{\sqrt{2}}(a;b)~|~a_i,b_i\in\ZZ,~\sum_{i=1}^4 a_i\in 2\ZZ,~ a_i-b_i\equiv a_j-b_j\mod2~~\forall~~i,j\in\lbrace1,2,3,4\rbrace \right\rbrace~.
\end{equation}
Reflecting $\Gamma^{4,4}_{\text{w-m}}$ amounts to changing the signature from $(4,4)$ to $(8,0)$, by sending $(a;b)\rightarrow(a|b)$ for all lattice vectors. This results precisely in the lattice $\Gamma^8_{\text{w-m}}$   which according to Lemma \ref{lattice_lem} is simply the $E_8$ lattice. This finishes the proof.   
\end{proof}

\section*{Acknowledgements}
We thank Shamit Kachru, Sarah Harrison, Theo Johnson-Freyd, Sander Mack-Crane, and Shu-Heng Shao for useful discussions. The work of M.C. and V.A. is supported by ERC starting grant H2020 \# 640159 and NWO vidi grant (number 016.Vidi.189.182). J.D. acknowledges support from the U.S. National Science Foundation (DMS 1203162, DMS 1601306), and the Simons Foundation (\#316779). M. C. also grately acknowledges the hospitality of the Mathematics Institute of Academica Sinica, as well as the National Center for Theoretical Sciences (NCTS) of Taiwan.

\appendix

\section{Sigma Model Symmetries} 
\label{app:sym}

\begin{table}[h]
	\begin{center}
	\footnotesize
		\begin{tabular}{|c|cccc||c|cccc||c|c|c|}
			\hline
			Class $\rho_{e}$ & 
			\multicolumn{4}{c||}{Non-trivial eigenv. in $\rho_{e}$}
			& Class $\rho_\psi$ &  	\multicolumn{4}{c||}{Eigenv. in $\rho_{\psi}$ (twice each)} & $o(g)$ & orb & $(E_8)^{\rho_e(g)}$\\[3pt]
			\hline
			$1A$ & 
			- & -&  -& -& $1A$ &$ 1$ & $ 1$ &   $ 1$  & $ 1$ & 1 & $T^4$ & rk$>4$\\
			$-1A$ & 
			- & -&  -& -& $-1A$ &$- 1$ & $- 1$ &   $- 1$  & $- 1$ & 2 & $K3$ & rk$>4$\\[3pt]
			$2B$
			& - & - & $-1$ & $-1$ & $2.2C$ & $\text{e}(\frac{1}{4})$ & $\text{e}(\frac{3}{4})$ & $\text{e}(\frac{1}{4})$ & $\text{e}(\frac{3}{4})$ & 4 & $K3$ & rk$>4$\\[3pt]
			$3A$
			& - & - & $\text{e}(\frac{1}{3})$ & $\text{e}(\frac{2}{3})$ & $3BC$ &$\text{e}(\frac{1}{3})$ & $\text{e}(\frac{2}{3})$ &  $\text{e}(\frac{1}{3})$ & $\text{e}(\frac{2}{3})$ & 3 & $K3$ & rk$>4$\\[3pt]
			$-3A$ 
			& - & - & $\text{e}(\frac{1}{3})$ & $\text{e}(\frac{2}{3})$ & $-3BC$ & $\text{e}(\frac{1}{6})$ & $\text{e}(\frac{5}{6})$ &  $\text{e}(\frac{1}{6})$ & $\text{e}(\frac{5}{6})$ & 6 & $K3$ & rk$>4$\\[3pt]
			$2A$ & 
			$-1$ & $-1$ & $-1$ & $-1$ & $2A$ &$ 1$ & $ 1$ & $-1 $& $- 1$ & 2 & $T^4$& $D_4$\\[3pt]
			$-2A$ & 
			$-1$ & $-1$ & $-1$ & $-1$ & $2A'$ & $- 1$ & $- 1$ & $ 1$& $ 1$ & 2 & $T^4$& $D_4$\\[3pt]
			$2E$ & 
			$-1$ & $-1$ & $-1$ & $-1$ & $2E$ &  $ 1$ & $ 1$ & $-1 $& $- 1$ & 2 & $T^4$& $A_1^4$\\[3pt]
			$-2E$ & 
			$-1$ & $-1$ & $-1$ & $-1$ & $2E'$ &  $- 1$ & $- 1$ & $ 1$& $ 1$  & 2 & $T^4$& $A_1^4$\\[3pt]
			$3E$ & 
			$\ex(\frac{1}{3})$ & $\text{e}(\frac{2}{3})$  & $\text{e}(\frac{1}{3})$& $\text{e}(\frac{2}{3})$ & $3E$ &  $ 1$ & $ 1$ &$\text{e}(\frac{1}{3})$ &  $\text{e}(\frac{2}{3})$  & 3 & $T^4$& $A_2^2$\\[3pt]
			$3E'$ & 
			$\text{e}(\frac{1}{3})$ & $\text{e}(\frac{2}{3})$  & $\text{e}(\frac{1}{3})$& $\text{e}(\frac{2}{3})$  & $3E'$ &$\text{e}(\frac{1}{3})$ &  $\text{e}(\frac{2}{3})$ & $ 1$ & $ 1$ & 3 & $T^4$& $A_2^2$\\[3pt]
			$-3E$ & 
			$\text{e}(\frac{1}{3})$ & $\text{e}(\frac{2}{3})$  & $\text{e}(\frac{1}{3})$& $\text{e}(\frac{2}{3})$  & $-3E$ &$ -1$ & $ -1$ & $\text{e}(\frac{1}{6})$ &  $\text{e}(\frac{5}{6})$ &  6 & $K3$ & $A_2^2$\\[3pt]
			$-3E'$ & 
			$\text{e}(\frac{1}{3})$ & $\text{e}(\frac{2}{3})$  & $\text{e}(\frac{1}{3})$& $\text{e}(\frac{2}{3})$  & $-3E'$ &$\text{e}(\frac{1}{6})$ &  $\text{e}(\frac{5}{6})$ & $ -1$ & $ -1$ & 6 & $K3$ & $A_2^2$\\[3pt]
			$4A$ 
			& $\text{e}(\frac{1}{4})$ & $\text{e}(\frac{3}{4})$  & $\text{e}(\frac{1}{4})$  & $\text{e}(\frac{3}{4})$& $4A$ & $ 1$ & $ 1$ & $\text{e}(\frac{1}{4})$ & $\text{e}(\frac{3}{4})$ &  4 & $T^4$& $D_4$\\[3pt]
			$4A'$ 
			& $\text{e}(\frac{1}{4})$ & $\text{e}(\frac{3}{4})$  & $\text{e}(\frac{1}{4})$ & $\text{e}(\frac{3}{4})$ & $4A'$ &$\text{e}(\frac{1}{4})$ & $\text{e}(\frac{3}{4})$ & $ 1$ & $ 1$ & 4 & $T^4$& $D_4$\\[3pt]
			$-4A$ 
			& $\text{e}(\frac{1}{4})$ & $\text{e}(\frac{3}{4})$  & $\text{e}(\frac{1}{4})$ & $\text{e}(\frac{3}{4})$ & $-4A$ &$- 1$ & $- 1$ & $\text{e}(\frac{1}{4})$ & $\text{e}(\frac{3}{4})$ &  4 & $K3$ & $D_4$\\[3pt]
			$-4A'$ 
			& $\text{e}(\frac{1}{4})$ & $\text{e}(\frac{3}{4})$  & $\text{e}(\frac{1}{4})$ & $\text{e}(\frac{3}{4})$ & $-4A'$ &$\text{e}(\frac{1}{4})$ & $\text{e}(\frac{3}{4})$ & $- 1$ & $ -1$ & 4 & $K3$ & $D_4$\\[3pt]
			$4C$ 
			& $\text{e}(\frac{1}{4})$ & $\text{e}(\frac{3}{4})$ & $-1$ & $-1$ & $8A$ &$\text{e}(\frac{1}{8})$& $\text{e}(\frac{7}{8})$ & $\text{e}(\frac{3}{8})$ &$\text{e}(\frac{5}{8})$ & 8 & $K3$ & $A_1A_3$\\[3pt]
			$-4C$ 
			& $\text{e}(\frac{1}{4})$ & $\text{e}(\frac{3}{4})$ & $-1$ & $-1$ & $-8A$ &  $\text{e}(\frac{3}{8})$ &$\text{e}(\frac{5}{8})$ &$\text{e}(\frac{1}{8})$& $\text{e}(\frac{7}{8})$ & 8& $K3$ & $A_1A_3$\\[3pt]
			$5A$ & 
			$\text{e}(\frac{1}{5})$ & $\text{e}(\frac{4}{5})$ & $\text{e}(\frac{2}{5})$ & $\text{e}(\frac{3}{5})$ & $5BC$ & $\text{e}(\frac{1}{5})$ & $\text{e}(\frac{4}{5})$ & $\text{e}(\frac{2}{5})$ & $\text{e}(\frac{3}{5})$ & 5 & $K3$ & $A_4$\\[3pt]
			$5A'$ & 
			$\text{e}(\frac{1}{5})$ & $\text{e}(\frac{4}{5})$ & $\text{e}(\frac{2}{5})$ & $\text{e}(\frac{3}{5})$ & $5BC'$ &$\text{e}(\frac{2}{5})$ & $\text{e}(\frac{3}{5})$ & $\text{e}(\frac{1}{5})$ & $\text{e}(\frac{4}{5})$  & 5 & $K3$ & $A_4$\\[3pt]
			$-5A$ & 
			$\text{e}(\frac{1}{5})$ & $\text{e}(\frac{4}{5})$ & $\text{e}(\frac{2}{5})$ & $\text{e}(\frac{3}{5})$ & $-5BC$ &$\text{e}(\frac{3}{10})$ & $\text{e}(\frac{7}{10})$ & $\text{e}(\frac{1}{10})$ & $\text{e}(\frac{9}{10})$ & 10 & $K3$ & $A_4$\\[3pt]
			$-5A$ & 
		$\text{e}(\frac{1}{5})$ & $\text{e}(\frac{4}{5})$ & $\text{e}(\frac{2}{5})$ & $\text{e}(\frac{3}{5})$ &  $-5BC'$ &$\text{e}(\frac{1}{10})$ & $\text{e}(\frac{9}{10})$ &$\text{e}(\frac{3}{10})$ & $\text{e}(\frac{7}{10})$ & 10 & $K3$ & $A_4$\\[3pt]
			$6A$ & 
			$\text{e}(\frac{1}{6})$ & $\text{e}(\frac{5}{6})$ & $-1$ & $-1$ & $6BC$ &$\text{e}(\frac{1}{3})$ & $\text{e}(\frac{2}{3})$ & $\text{e}(\frac{1}{6})$ & $\text{e}(\frac{5}{6})$ & 6 & $K3$ & $D_4$\\[3pt]
			$-6A$ & 
			$\text{e}(\frac{1}{6})$ & $\text{e}(\frac{5}{6})$ & $-1$ & $-1$ & $6BC'$ &$\text{e}(\frac{1}{6})$ & $\text{e}(\frac{5}{6})$ &$\text{e}(\frac{1}{3})$ & $\text{e}(\frac{2}{3})$ &  6 & $K3$ & $D_4$\\[3pt]
			$6D$ & 
			$\text{e}(\frac{1}{3})$ & $\text{e}(\frac{2}{3})$ & $-1$ & $-1$ & $12BC$ &$\text{e}(\frac{1}{12})$ & $\text{e}(\frac{11}{12})$ & $\text{e}(\frac{5}{12})$ & $\text{e}(\frac{7}{12})$ & 12 & $K3$ & $A_1^2A_2$\\[3pt]
			$-6D$ & 
			$\text{e}(\frac{1}{3})$ & $\text{e}(\frac{2}{3})$ & $-1$ & $-1$ & $-12BC'$ &$\text{e}(\frac{5}{12})$ & $\text{e}(\frac{7}{12})$ & $\text{e}(\frac{1}{12})$ & $\text{e}(\frac{11}{12})$ &  12 & $K3$ & $A_1^2A_2$\\[3pt]
			\hline
		\end{tabular}
	\end{center}
	\caption{
	}\label{t:classes}
\end{table}

In this appendix we record the cyclic symmetry subgroups of  four-torus sigma models. 
Given that $G_1 < O^+_8(2)$ and $G_0<W^+(E_8)$, we require to consider the lifts of relevant classes $X$ of $O_8^+(2)$ to $W^+(E_8)$. See \eqref{diagram}.
If there are two classes in the lift, they are denoted $\pm X$. We use the notation $2.2C$ to refer to the lift of the class $2C\subset O_8^+(2)$ to $W^+(E_8)$, which is a single class of order 4 rather than two classes $\pm 2C$. We follow \cite{ATLAS} for the naming of the classes.

Note that the set of possible $G_1$ is bijective to the set of subgroups of $W^+(E_8)$ which fix an $E_8$-sublattice of rank at least four, since there is always a rank four subspace in the representation $\rho_{\rm e}$ in $G_0$. The column ``non-trivial eigenvalues in $\rho_{\rm e}$" 
records the non-trivial eigenvalues in each case. 
Correspondingly, the $W^+(E_8)$ classes $\pm X$ in the columns ``Class $\rho_{\rm e}$'' denotes the preimage of the class $X\subset O^+_8(2)$ under the projection $\pi'$ of \eqref{diagram}.

In \S\ref{sec:symmetry groups} we have learned that this is not the only way to obtain a lift of a class of $ O^+_8(2)$ in the context of four-torus sigma models. In the column ``Class $\rho_{\psi}$'' we record the  preimage of the class $X\subset O^+_8(2)$ under the projection $\pi\mydprime$  in \eqref{diagram}. Note that the ``Class $\rho_{\psi}$'' and ``Class $\rho_{\rm e}$'', are of course related by a triality transformation which exchanges $\iota_s$ and $\iota_v$, and correspondingly $\rho_{\psi}$ and $\rho_{\rm e}$. By \eqref{def:gact_chi}, each eigenvalue appears twice in  $\rho_{\psi}$ and we therefore group the eight eigenvalues in four pairs (of identical values) and record just representative eigenvalues for each of these pairs. In the notation of  \eqref{def:gact_chi}, the first two eigenvalues are $\zeta_L$ and $\zeta_L^{-1}$ while the last two are  $\zeta_R$ and $\zeta_R^{-1}$.
The notation $\pm X'$ is a reminder that, the same $W^+(E_8)$ class can act differently on a four-torus sigma model by exchanging left- and right-movers. 

In the last part of Table \ref{t:classes} we write $o(g)$ for the order of the element in $G_0$ (i.e. in the faithful representation $\rho_\psi$), while the order in $G_1=G_0/\ZZ_2$ (i.e. in the unfaithful representation $\rho_{\rm e}$,) can be read off from the symbol of the class, since $G_1< O_8^+(2)$. We also indicate whether the orbifold by $g$ is a sigma model on $T^4$ or $K3$. 
Finally, we indicate the $\rho_{e}(g)$-fixed sublattice of $E_8$ if it has rank four, in which case the symmetry $g$ is non-geometric and appears only at a single point in the moduli space characterized by the fixed sublattice, which we record. If the rank is larger than four then the symmetry is geometric and it occurs in some family of models. 

\section{Cocycles and Lifts}
\label{app:cocycles}

In this appendix, we review some well-known results about the  OPE of vertex operators in toroidal sigma models and in lattice vertex operator algebras, with a particular focus on the so called `cocycle factors'. Some early references on the subject are \cite{Borcherds86,FLM} in the VOA literature and \cite{Gross:1985fr} in string theory; further references include \cite{DolanGoddardMontague,holomvoas,Borcherds}. In this section, we adopt the language of two dimensional conformal field theory: the lattice VOA version of our statements can be easily derived from the particular case of chiral CFTs.

Let us consider a (bosonic) toroidal conformal field theory, describing $d_+$ chiral and $d_-$ anti-chiral compact free bosons, whose discrete winding-momentum (Narain) lattice is an even unimodular lattice $L$ of dimension $d=d_++d_-$, whose bilinear form $(\cdot,\cdot):L\times L\to \ZZ$ has signature $(d_+,d_-)$. Note that such a lattice exists only when $d_+-d_-\equiv 0\mod 8$. If $d_-=0$, then the conformal field theory is chiral, and it can be described as a lattice vertex operator algebra based on the even unimodular lattice $L$. On the opposite extreme, if $d_+=d_-=d/2$, the CFT can be interpreted as a sigma model on a torus $T^{d/2}$. The supersymmetric versions of these models are obtained by adjoining $d^+$ chiral and $d^-$ anti-chiral free fermions. The properties we are going to discuss do not depend on whether the toroidal CFT is bosonic or supersymmetric, so we will focus on the bosonic case for simplicity. As discussed in \S\ref{subsec:sym_T4}, for a given unimodular lattice $L$, there is a whole moduli space of toroidal models based on $L$, whose points correspond to different decompositions $L\otimes \RR=\Pi_L\oplus\Pi_R$ into a positive definite subspace $\Pi_L$ and a negative definite one $\Pi_R$. Every vector $v\in L\otimes \RR$ can be decomposed accordingly as $v=(v_L,v_R)$. We can define positive definite scalar products on $\Pi_L$ and on $\Pi_R$, that are uniquely determined by the condition
\be (\lambda,\mu)=\lambda_L\cdot \mu_L-\lambda_R\cdot \mu_R\ ,
\ee for all $\lambda,\mu\in L\otimes \RR$.

The CFT contains the vertex operators $V_\lambda(z,\bar z)$, for each $\lambda\in L$, with OPE satisfying
\be\label{OPEcoc} V_\lambda(z,\bar z)V_\mu(w,\bar w)=\epsilon(\lambda,\mu) (z-w)^{\lambda_L\cdot \mu_L}(\bar z-\bar w)^{\lambda_R\cdot \mu_R} V_{\lambda+\mu}(w,\bar w)+\ldots
\ee where $\ldots$ are subleading (but potentially still singular) terms. In the chiral ($d_-=0$) case, one can simply set $\lambda_L=\lambda$ and $\lambda_R=0$ and similarly with $\mu$.
Here, $\epsilon:L\times L\to U(1)$ must satisfy
\begin{align}
&\epsilon(\lambda,\mu)=(-1)^{\bil{\lambda}{\mu}}\epsilon(\mu,\lambda)\label{cohclass}\\
&\epsilon(\lambda,\mu)\epsilon(\lambda+\mu,\nu)=\epsilon(\lambda,\mu+\nu)\epsilon(\mu,\nu)\qquad \qquad \text{(cocycle condition)}
\label{cocycle}\end{align}
in order for the OPE to be local and associative. Given a solution $\epsilon(\lambda,\mu)$ to these conditions, any other solution is given  by
\be \tilde \epsilon(\lambda,\mu)= \epsilon(\lambda,\mu)\frac{b(\lambda)b(\mu)}{b(\lambda+\mu)}\ ,
\ee for an arbitrary $b:L\to U(1)$. This change corresponds to a redefinition of the fields $V_\lambda$: if $V_\lambda(z,\bar z)$ obey the OPE \eqref{OPEcoc} with cocyle $\epsilon$, then the operators $\tilde V_\lambda(z,\bar z)=b(\lambda)V_\lambda(z,\bar z)$ obey \eqref{OPEcoc} with the cocycle $\tilde \epsilon$. Notice that if $b(\lambda+\mu)=b(\lambda)b(\mu)$ for all $\lambda,\mu\in L$ (i.e. if $b:L\to U(1)$ is a homomorphism of abelian groups), then $\epsilon$ is unchanged, and the transformation $V_\lambda(z,\bar z)\to b(\lambda)V_\lambda(z,\bar z)$ is a symmetry of the CFT, which is part of the $U(1)^d$ group generated by the zero modes of the currents. 

One can show that $\epsilon(\lambda,\mu)$ satisfying the conditions \eqref{cohclass} and \eqref{cocycle} can be chosen to take values in $\{\pm 1\}$. Furthermore, one can use the freedom in redefining $V_\lambda$ to set
\be \epsilon(0,\lambda)=\epsilon(\lambda,0)=1,\qquad\qquad \forall\lambda\in L,
\ee so that $V_0(z,\bar z)=1$. Cocycles satisfying this condition are sometimes called normalized.  Finally, one can choose $\epsilon$ such that\footnote{One further condition that is usually imposed is
	$\epsilon(-\lambda,\lambda)=1$ for all $\lambda\in L$. With this choice the general relation $(V_\lambda)^\dag=\epsilon(\lambda,-\lambda)V_{-\lambda}$ simplifies as $(V_\lambda)^\dag=V_{-\lambda}$. Another common choice is $\epsilon(-\lambda,\lambda)=(-1)^{\lambda^2/2}$. We will not impose any of these conditions.}
\be \epsilon(\lambda+2\nu,\mu)=\epsilon(\lambda,\mu+2\nu)=\epsilon(\lambda,\mu)\ , \qquad \forall \lambda,\mu,\nu\in L\ .
\ee If we require all these conditions, then $\epsilon$ determines a well defined function $L/2L\times L/2L\to \{\pm 1\}$.

More formally (see for example \cite{FLM}), the cocycle $\epsilon$ represents a class in the cohomology group $H^2(L,\ZZ/2\ZZ)$, where the lattice $L$ is simply regarded as an abelian group. These cohomology classes are in one to one correspondence with isomorphism classes of central extensions 
$$ 1\to \ZZ/2\ZZ\to \hat L\to L\to 1\ ,
$$ of the abelian group $L$ by $\ZZ/2\ZZ$. The specific cohomology class that is relevant for the toroidal CFT is uniquely determined by the condition \eqref{cohclass}. Using this formalism, the CFT can alternatively be defined by introducing a vertex operator $V_{\hat\lambda}$ for each element $\hat\lambda\in \hat L$ in this central extension. Then, the OPE of $V_{\hat\lambda}(z,\bar z)V_{\hat\mu}(w,\bar w)$ is analogous to \eqref{OPEcoc}, with $\epsilon(\lambda,\mu) V_{\lambda+\mu}$ replaced by $V_{\hat \lambda\cdot\hat \mu}$ (here, $\hat \lambda\cdot\hat \mu$ denotes the composition law in the extension $\hat L$, which is possibly non-abelian). Our previous description of the CFT can be recovered  by choosing a section $e:L\to \hat L$ and defining the vertex operators $V_\lambda:=V_{e(\lambda)}$ for each $\lambda\in L$. This leads to the OPE \eqref{OPEcoc}, where the particular cocycle representative $\epsilon$ depends on the choice of the section $e$ via $e(\lambda)e(\mu)=\epsilon(\lambda,\mu)e(\lambda+\mu)$.

\bigskip

An automorphism $g\in O(L)$ can be lifted (non-uniquely) to a symmetry $\hat g$ of the CFT such that
\be \hat g(V_\lambda (z,\bar z))=\xi_g(\lambda)V_{g(\lambda)}(z,\bar z)\ ,
\ee where $\xi_g:L\to U(1)$ must satisfy
\be\label{xigcond} \frac{\xi_g(\lambda)\xi_g(\mu)}{\xi_g(\lambda+\mu)}=\frac{\epsilon(\lambda,\mu)}{\epsilon(g(\lambda),g(\mu))}\ .
\ee As shown below, $\xi_g$ satisfying this condition always exists, and any two such $\xi_g,\tilde \xi_g$ are related by $\tilde \xi_g(\lambda)=\rho(g)\xi_g(\lambda)$, where $\rho:L\to U(1)$ is a homomorphism.  Furthermore, one can always find $\xi_g$ taking values in $\{\pm 1\}$ and such that
\begin{align} &\xi_g(0)=1\\ & \xi_g(\lambda+2\mu)=\xi_g(\lambda) \qquad \forall \lambda,\mu\in L\ .
\end{align} 
With these condition, $\xi_g$ induces a well-defined map $\xi_g:L/2L\to \{\pm 1\}$.

A constructive proof of these statements is as follows (see \cite{DolanGoddardMontague}). Choose a basis $e_1,\ldots, e_d$  for $L$. Define an algebra of operators $\gamma_i\equiv \gamma_{e_i}$, $i=1,\ldots,d$, satisfying\footnote{A slightly modified definition sets $\gamma_i^2=(-1)^{e_i^2/2}$. With the latter choice, one obtains $\epsilon(\lambda,-\lambda)=(-1)^{\lambda^2/2}$ for all $\lambda\in L$, and $\gamma_\lambda$ depends on $\lambda\mod 4L$ rather than $2L$. However, both $\epsilon$ and $\xi_g$ are still well defined on $L/2L$.}
\be \gamma_i^2=1\qquad \gamma_i\gamma_j=(-1)^{\bil{e_i}{e_j}} \gamma_j\gamma_i\ ,
\ee and for every $\lambda=\sum_{i=1}^d a_i e_i\in L$, set
\be \gamma_\lambda:=\gamma_1^{a_1}\cdots \gamma_d^{a_d}\ .
\ee Then, the following properties hold:
\be \gamma_0=1\qquad \qquad  \gamma_{\lambda+2\mu}=\gamma_{\lambda}\qquad \qquad \gamma_\lambda\gamma_\mu=(-1)^{\bil{\lambda}{ \mu}}\gamma_\mu\gamma_\lambda\ .
\ee Define $\epsilon:L\times L\to \{\pm1\}$ by
\be \gamma_\lambda\gamma_\mu=\epsilon(\lambda,\mu)\gamma_{\lambda+\mu}\ ,
\ee and, for every $g\in O(L)$, define $\xi_g:L\to \{\pm 1\}$ by
\be \gamma_{g(\lambda)}=\xi_g(\lambda)\gamma_{g(e_1)}^{a_1}\cdots \gamma_{g(e_d)}^{a_d}\ .
\ee  It is easy to verify that $\epsilon$ and $\xi_g$ satisfy all the properties mentioned above. In particular, this choice of $\xi_g$ is such that $\xi_g(e_i)=1$ for all the basis elements $e_i$.
It is clear that $\gamma_\lambda$, and therefore also $\epsilon$ and $\xi_g$, depend on $\lambda$ only mod $2L$.

The constraints that we imposed on $\xi_g$ still leave some freedom in the choice of the lift. There are two further conditions that one might want to impose:
\begin{itemize} \item[(A)] One might require $\hat g$ to have the same order $N=|g|<\infty$ as $g$. Notice that if $\hat g$ is a lift of a $g$ of order $N$, then
	\be \hat g^N(V_\lambda)=\xi_g(\lambda)\xi_g(g(\lambda))\cdots \xi_g(g^{N-1}(\lambda))V_\lambda\ ,
	\ee so that $\hat g^N=1$ if and only if
	\be \xi_g(\lambda)\xi_g(g(\lambda))\cdots \xi_g(g^{N-1}(\lambda))=1\qquad \forall \lambda\in L\ .
	\ee 
	\item[(B)] Alternatively, one might want $\xi_g(\lambda)$ to be trivial whenever $\lambda$ is $g$-fixed
	\be \xi_g(\lambda)=1\qquad \forall \lambda\in L^g\ ,
	\ee or, equivalently,
	\be \hat g(V_\lambda)=V_\lambda \qquad \forall \lambda\in L^g\ .
	\ee Lifts satisfying this property are usually called \emph{standard lifts}.
\end{itemize}

\begin{proposition}
	Every $g\in O(L)$ admits a \emph{standard lift} $\hat g$, i.e. such that $\hat g(V_\lambda)=V_\lambda$  for all $\lambda\in L^g$.
\end{proposition} 
\begin{proof}
	For all $\lambda, \mu\in L^g$, one has obviously $\frac{\epsilon(g(\lambda),g(\mu))}{\epsilon(\lambda,\mu)}=1$. Therefore, the restriction of $\xi_g$ to $L^g$ is a homomorphism $L^g\to \{\pm 1\}$, and it is trivial if and only if it is trivial on all elements of a basis of $L^g$. By the construction described above, one can always find a lift $\hat g$ such that $\xi_g$ is trivial for all the elements of a given basis of $L$. Choose a basis of $L^g$; since $L^g$ is primitive in $L$, this can be completed to a basis of $L$.  By choosing $\xi_g$ to be trivial on the elements of this basis, we obtain a lift $\hat g$ satisfying condition (B).
\end{proof}

Standard lifts are not unique, but they are all conjugate to one each other within the symmetry group of the CFT, as the following proposition shows. (The following two propositions are proved in \cite{holomvoas}.)

\begin{proposition}\label{standconj}
	Let $g \in O(L)$ and $\hat g$, $\hat g'$ be two lifts of $g$ with associated functions $\xi_g,\xi_g': L \to \{\pm 1\}$. Suppose $\xi_g=\xi_g'$ on the fixed-point sublattice $L^g$. Then $\hat g$ and $\hat g'$ are conjugate in the group of symmetries of the CFT.
\end{proposition}

Since the order and the twined genus of a lift $\hat g$ depends only on its conjugacy class within the group of symmetries, this proposition then tells us that these quantities only depend on the restriction of $\xi_g$ on the fixed sublattice $L^g$. In particular, when $g$ fixes no sublattice of $L$, all its lifts $\hat g$ are conjugate to each other.

The following result gives, for the standard lifts (i.e. for $\xi_g=1$ on $L^g$), the order of $\hat g$ and the action of every power $\hat g^k$ on the corresponding $g^k$-fixed sublattice $L^{g^k}$ 

\begin{proposition}\label{ordstand}
	Let $g\in O(L)$ and $\hat g$ be a standard lift (i.e. $\xi_g(\lambda)=1$ for all $\lambda \in L^g$). Then:
	\begin{enumerate}
		\item 	If $g$ has odd order $N$, then $\hat g^k(V_\lambda)=V_\lambda$
		for all $\lambda\in L^{g^k}$. In particular $\hat g$ has order $N$.
		\item 	If $g$ has even order $N$, then for all $\lambda \in  L^{g^k}$, 
		\be {\hat{g}}^k(V_\lambda) = \begin{cases} V_\lambda & \text{for $k$ odd,}\\
			(-1)^{\bil{\lambda}{g^{k/2}(\lambda)}}V_\lambda & \text{for $k$ even.}
		\end{cases}
		\ee 
		In particular $\hat g$ has order $N$ if $\bil{\lambda}{g^{N/2}(\lambda)}$ is even for all $\lambda\in L$ and order $2N$ otherwise.
	\end{enumerate}
\end{proposition}

For practical applications of this proposition it is important to have an easy way to determine if $\bil{\lambda}{g^{N/2}(\lambda)}$ is even for all $\lambda\in L$. Consider $g$ of order $2$ (these are the important cases, since $g^{N/2}$ is always of order $2$). One has \be \bil{\lambda}{g(\lambda)}\equiv \frac{1}{2}(\lambda+g(\lambda))^2\equiv 2 \bigl(\frac{1+g}{2}(\lambda)\bigr)^2\mod 2\ .\ee Since $\frac{1+g}{2}$ is the projector onto the $g$-invariant subspace  $L^g\otimes \RR$ of $L\otimes \RR$,  by self-duality of $L$, one has $\frac{1+g}{2}(L)=(L^g)^*$. Therefore, the existence of $\lambda\in L$ with $\bil{\lambda}{g(\lambda)}$ odd is equivalent to the existence of $v\in (L^g)^*$ with half-integral square norm $v^2\in \frac{1}{2}+\ZZ$. This condition is quite easy to check, once the lattice $L^g$ is known. When the fixed sublattice $L^g$ is positive definite, the order of the standard lift can also be related to properties of the lattice theta series  $\theta_{L^g}(\tau)=\sum_{\lambda\in L^g} q^{\lambda^2/2}$. This is well known to be a modular form of weight $r/2$, where $r$ is the rank of $L^g$, for a congruence subgroup of $SL_2(\ZZ)$. Its S-transform $\theta_{L^g}(-1/\tau)$ is proportional to the theta series $\theta_{(L^g)^*}(\tau)$ of the dual lattice $(L^g)^*$. If $(L^g)^*$ contains a vector $v$ with half-integral square norm $v^2\in \frac{1}{2}+\ZZ$, then the $q$-series of $\theta_{(L^g)^*}(\tau)=\sum_{v\in (L^g)^*} q^{\frac{v^2}{2}}$ contains some powers $q^n$ with $n\in \frac{1}{4}\ZZ$. As a consequence, the standard lift of $g$ of order $2$ has order $2$ if and only if the theta series $\theta_{L^g}(\tau)$ is a modular form for a subgroup of level $2$, while it has order $4$ if it is only modular under a subgroup  of $SL_2(\ZZ)$ of level $4$.

\bigskip

When $g$ has even order $N$ and its standard lift $\hat g$ has order $2N$, it is sometimes convenient to choose a non-standard lift $\hat g$ with the same order $N$ as $g$.  The next proposition shows that for $N=2$ such a lift always exists.

\begin{proposition}
	Let $g\in O(L)$ have order $2$. Then,  there is a lift $\hat g$ of $g$ of order $2$.
\end{proposition}
\begin{proof}
	Let $\hat g'$ be a standard lift of $g$. If $\bil{\lambda}{ g(\lambda)}$ is even for all $\lambda\in L$, then by the previous proposition $\hat g'$ has order $2$ and we can just set $\hat g=\hat g'$. Suppose that $\bil{\lambda}{ g(\lambda)}$ is odd for some $\lambda\in L$.
	One has $(-1)^{\bil{\lambda}{ g(\lambda)}}=(-1)^{\frac{(\lambda+g(\lambda))^2}{2}}$, and the map $\lambda+g(\lambda)\mapsto (-1)^{\frac{(\lambda+g(\lambda))^2}{2}}$ is a homomorphism $(1+g)L\to \{\pm 1\}$. Thus, there is $w\in ((1+g)L)^*$ such that $(-1)^{\frac{(\lambda+g(\lambda))^2}{2}}=(-1)^{w\cdot (\lambda+g(\lambda))}$ for all $\lambda\in L$. Notice that $(1+g)L\subseteq L^g$, so that $(L^g)^*\subseteq ((1+g)L)^*$. On the other hand, it is easy to see that $w\in (L^g)^*$, i.e. that $\bil{v}{w}\in \ZZ$ for all $v\in L^g$. Indeed, if $v\in L^g$, then either $v\in (1+g)L$ (in which case, $\bil{v}{w}\in \ZZ$ is obvious) or $2v\in (1+g)L$ (because $2v=v+g(v)$ for $v\in L^g$). In the latter case. one has $(-1)^{\bil{2v}{w}}=(-1)^{\frac{(2v)^2}{2}}=1$, so that $\bil{2v}{w}$ must be even, and therefore $\bil{v}{w}\in \ZZ$. Finally, by self-duality of $L$, for every  $w\in (L^g)^*$  there always exist $\tilde w\in L$ such that $\bil{\tilde w}{v}=\bil{w}{v}$ for all $v\in L^g$. In particular, $(-1)^{\bil{\tilde w}{\lambda+g(\lambda)}}=(-1)^{\bil{\lambda}{g(\lambda)}}$ for all $\lambda\in L$. Then, we can define the lift $\hat g$ by $\xi_g(\lambda)=\xi'_g(\lambda)(-1)^{\bil{\tilde w}{\lambda}}$, where $\xi'_g$ is the function corresponding to a standard lift. Thus, for all $\lambda\in L$,
	\begin{align} \hat g^2(V_\lambda)&=\xi_{g}(\lambda)\xi_g(g(\lambda))V_\lambda=\xi_g(\lambda+g(\lambda))\frac{\epsilon(\lambda,g(\lambda))}{\epsilon(g(\lambda),\lambda)}V_\lambda\\
	&=\xi'_g(\lambda+g(\lambda))(-1)^{\bil{\tilde w}{\lambda+g(\lambda)}}(-1)^{\bil{\lambda}{g(\lambda)}}V_\lambda=V_\lambda\ ,
	\end{align} where we used the condition \eqref{xigcond}, and the fact that $\xi'_g(\lambda+g(\lambda))=1$, since $\lambda+g(\lambda)\in L^g$ and $\hat g'$ is a standard lift. We conclude that $\hat g$ has order $2$.
\end{proof}

\subsection{Applications}

Let us now apply the results described in the previous section to the cases we are interested in, namely the sigma model on $T^4$ and the SVOA based on the $E_8$ lattice. As explained in the article, there is a correspondence between automorphisms $g$ of the lattice $\Gamma^{4,4}$ lifting to symmetries that preserve the $\N=(4,4)$ superconformal algebra, and certain automorphisms of the lattice $E_8$. One needs to choose a lift of these lattice automorphisms to symmetries of the corresponding conformal field theory or SVOA. As explained above, a lift is determined, up to conjugation by CFT symmetries, by the restriction of the function $\xi_g$ to the $g$-fixed sublattice. The most obvious choice is to consider the standard lift both for the sigma model and for the SVOA, so that $\xi_g$ is trivial on the fixed sublattices. In general, the order of the standard lift is either the same or twice the order of the lattice automorphism. Therefore, it is not obvious a priori that the standard lifts in the sigma model and in the SVOA have the same order; we will show now that this is always true in the present the case.

Let $g$  be an automorphism of the lattice $\Gamma^{4,4}$. We denote any such automorphism by the class of $\rho_\psi$, as in Table \ref{t:classes}. Using Propositions \ref{standconj} and \ref{ordstand}, the orders of the standard lifts are as follows.

\begin{itemize}
	\item Classes of odd order $N$ (1A, 3BC, 3E, 3E', 5BC, 5BC'): since $N$ is odd, the standard lift has also order $N$. This conclusion holds also for the lift of the corresponding automorphisms of the $E_8$ lattice.
	\item Class -1A: an automorphism $g$ in this class flips the sign of all vectors in $\Gamma^{4,4}$. Therefore, it acts trivially on $\Gamma^{4,4}/2\Gamma^{4,4}$, so that one can set $\xi_g(\lambda)=1$ for all $\lambda\in \Gamma^{4,4}$, and this lift has obviously order $2$. Since $g$ fixes no sublattice, any other lift of $g$ is conjugate to the lift above and has order $2$. This also implies that any lift $\hat g$ of a lattice automorphism $g$ of even order $N$, and such that $g^{N/2}$ is in class -1A, has order $N$. Indeed, $\hat g^{N/2}$ is a lift of a symmetry in class -1A, so that it must have order $2$. This argument applies to all $g$ in the classes 2.2C, -3BC, -3E, -3E', 8A, -8A, -5BC, -5BC', 12BC, -12BC'.
	An analogous reasoning holds for the automorphism of the lattice $E_8$ corresponding to class -1A, which flips the sign of all vectors in $E_8$. This automorphism has no fixed sublattice and acts trivially on $E_8/2E_8$, so that one can take $\xi_g$ to be trivial. The same reasoning as for the sigma model case shows that all lifts of this symmetry are conjugate to each other and have order $N=2$. More generally,  all automorphisms of $E_8$ in the classes 2.2C, -3BC, -3E, -3E', 8A, -8A, -5BC, -5BC', 12BC, -12BC' lift to symmetries of the SVOA of the same order.
	\item Classes 2A and 2A': the fixed sublattice is isomorphic to the root lattice $D_4$,  and its dual $D_4^*$ is an integral lattice. In particular,  $D_4^*$ contains no vector of half-integral square norm, and therefore the standard lift has order $2$. Furthermore, for any $g$ of even order $N$ such that $g^{N/2}$ is in class 2A or 2A', one has that $\bil{\lambda}{g^{N/2}(\lambda)}$ is even for all $\lambda$, so that a standard lift has the same order $N$. This applies to all $g$ in the classes 4A, 4A', -4A, -4A', 6BC, 6BC'. For automorphisms of the $E_8$ lattice in classes 2A and 2A', the fixed sublattice is also isomorphic to $D_4$, so the standard lift has the same order $N=2$. The same reasoning holds for the standard lifts of automorphisms in the classes 4A, 4A', -4A, -4A', 6BC, 6BC'.
	\item Classes 2E and 2E': the fixed sublattice is $A_1^4$, and its dual $(A_1^4)^*$ contains vectors of square length $1/2$. Thus, the standard lift has order $2N=4$. The corresponding automorphism of the $E_8$ lattice also fixes a sublattice isomorphic to $A_1^4$, so its standard lift has order $4$. 
\end{itemize}

The conclusion of this analysis is that, both for toroidal sigma models and for the $E_8$ SVOA, the only case where the standard lift has twice the order of the corresponding lattice automorphism is for the class 2E.

If $g$ is in class 2E, the twined genus for the standard lift (which has order $4$) involves the theta series of the $A_1^4$ lattice
\be \Theta_{A_1^4}(\tau)=\theta_3(2\tau)^4\ .
\ee This theta series (and the corresponding twined genus) is a modular form of level $4$. This is consistent with the analysis above.

One can also focus on a (non-standard) lift of order $2$, with $\xi_g(\lambda)=(-1)^{\lambda^2}/2$ for all $\lambda\in (1+g)\Gamma^{4,4}$. For any $g$ of order $2$, one has $(1+g)\Gamma^{4,4}=2((\Gamma^{4,4})^g)^*$; in particular, for $g$ in class 2E or 2E', one has $(\Gamma^{4,4})^g\cong A_1^4$, so that $(1+g)\Gamma^{4,4}\cong 2(A_1^4)^*\cong A_1^4\cong \Gamma^g$.  For this lift, the twining genus involves the theta series with characteristics
\be  \Theta_{A_1^4,\xi_g}(\tau)=\sum_{\lambda\in A_1^4} q^{\lambda^2/2} (-1)^{\lambda^2/2}=\Theta_{A_1^4}(\tau+\frac{1}{2})=\theta_3(2\tau+1)^4=\theta_4(2\tau)^4
\ee which is modular (with multipliers) for $\Gamma_0(2)$ (its S-transform is proportional to $\theta_4(\tau/2)^4$). As for the $E_8$ SVOA, since the sublattice fixed by the automorphism is also isomorphic to $A_4$, one can choose an analogous (non-standard) lift with the same $\xi_g$ on the fixed sublattice, which is also of order $2$.

For a general class, it is difficult to define a reasonable correspondence between non standard lifts in the sigma model and the $E_8$ SVOA, since the fixed sublattices are, in general, not isomorphic.

\section{The K3 Case Revisited}
\label{app:k3teg}

In \cite{derived} it was shown that the Conway group action on $\vsn$ may be used to recover many of the weak Jacobi forms that arise as twined elliptic genera of K3 sigma models. It was conjectured in op. cit. that all twined K3 elliptic genera arise in this way, but the analysis of \cite{umbral_k3} subsequently showed that there are four exceptions. In \S\ref{app:k3teg-teg} we explain how all but one of these exceptional cases may be recovered if we allow non-supersymmetry-preserving automorphisms of $\vsn$, and the remaining one too if we allow linear combinations of supersymmetry-preserving automorphisms of $\vsn$ with higher than expected order. 

VOSAs $\vfn$ and $\vsn$ are studied in \cite{Duncan07} and \cite{conway,derived}, respectively, in connection with moonshine for the Conway group. In \S\ref{app:k3teg-cms} we briefly review the relationship between these objects, 
and explain a sense in which the Conway group $\Co_0$ (see (\ref{eqn:Co0AutLL})) arises naturally as a group of automorphisms of the latter. Specifically, we introduce the notion of {Ramond (sector) ${\cal N}=1$ structure}, show that $\vsn$ admits such a structure, and demonstrate that $\Co_0$ is the full group of automorphisms of this structure. 
We also explain why $\vfn$ and $\vsn$ are the same as far as twinings of the K3 elliptic genus are concerned.

\subsection{Twined Elliptic Genera}\label{app:k3teg-teg}

We begin by reviewing the exceptional forms identified in \cite{umbral_k3}. Three of them actually arise in Mathieu moonshine, as the weak Jacobi forms associated to the conjugacy classes $3B$, $4C$ and $6B$ of $M_{24}$. (As before we adopt the notation of \cite{ATLAS} for conjugacy classes.) According to \cite{matmooreview}, for example, these forms are given respectively by
\begin{gather}
	\begin{split}\label{eqn:exceptionalforms1}
	Z_{3|3}(\tau,z)&=2\frac{\eta(\tau)^6}{\eta(3\tau)^2}\phi_{-2,1}(\tau,z),\\
	Z_{4|4}(\tau,z)&=2\frac{\eta(\tau)^4\eta(2\tau)^2}{\eta(4\tau)^2}\phi_{-2,1}(\tau,z),\\
	Z_{6|6}(\tau,z)&=2\frac{\eta(\tau)^2\eta(2\tau)^2\eta(3\tau)^2}{\eta(6\tau)^2}\phi_{-2,1}(\tau,z),
	\end{split}
\end{gather}
where 
$\phi_{-2,1}=-\theta_{1}^2\eta^{-6}$ is the unique weak Jacobi form of weight $-2$ and index $1$  for $\SL_2(\ZZ)$ such that $\phi_{-2,1}(\tau,z)=y^{-1}-2+y+O(q)$ for $q=e^{2\pi i\tau}$ and $y=e^{2\pi i z}$.
The subscripts $n|h$ in (\ref{eqn:exceptionalforms1}) encode the characters (i.e. multiplier systems) of the respective forms. See (1.9) and (3.8) of \cite{matmooreview}, for example, for the details of this. We denote the remaining exceptional form $Z_{8|4}$. It is given explicitly by
\begin{gather}\label{eqn:exceptionalforms2}
	Z_{8|4}(\tau,z)=2\frac{\eta(2\tau)^4\eta(4\tau)^2}{\eta(8\tau)^2}\phi_{-2,1}(\tau,z).
\end{gather}

Next we recall that in \S9 of \cite{derived} a holomorphic function $\phi_g(\tau,z):\HH\times \CC\to \CC$ is associated to each element $g$ of the Conway group 
\begin{gather}\label{eqn:Co0AutLL}
\Co_0:=\Aut(\LL)
\end{gather} 
such that the space of $g$-fixed points in $\LL\otimes_\ZZ \CC\simeq \CC^{24}$ is at least $4$-dimensional. In (\ref{eqn:Co0AutLL}) we write $\LL$ for the Leech lattice (cf. e.g. \cite{packings,ATLAS}). Now the full automorphism group of the VOSA structure on $\vsn$ is a $\ZZ_2$ quotient of the Lie group $\Spin_{24}(\CC)$, and we observe here that the construction of op. cit. works equally well for for any element of $\Spin_{24}(\CC)$ whose image in $\SO_{24}(\CC)$ fixes a $4$-space in $\LL\otimes_\ZZ\CC$. For example, consider an orthogonal transformation $x\in \SO_{24}(\CC)$ with Frame shape $\pi_x= 2^8.4^2$ (so that the characteristic polynomial of $x$ is $(1-x^2)^8(1-x^4)^2$). Then we have $C_{-x}=D_{x}=0$ in the notation of \cite{derived}, and a computation reveals that $\phi_{x}=Z_{4|4}$. Similarly we recover $Z_{6|6}$ and $Z_{8|4}$ by taking $\pi_x$ to be $2^6.6^2$ and $2^4.8^2$, respectively. 

The Frame shapes $2^8.4^2$, $2^6.6^2$, and $2^4.8^2$ are not represented by elements of the Conway group, and the Conway group is distinguished in that it arises as the stabilizer of any $\N=1$ structure on $\vsn$ (cf. \S\ref{app:k3teg-cms}). So such symmetries of $\vsn$ do not preserve supersymmetry, but it is notable that we can recover three of the four exceptional twined K3 elliptic genera by allowing these more general twinings on the VOSA side. 

Another interesting coincidence is that fact that $Z_{4|4}=\phi_g$, for $g\in \Co_0$ with Frame shape $\pi_g=2^4.4^{-4}.8^4$. (For this we take $D_g=16$, in the notation of \cite{derived}.) The surprising part is that $g$ has order $8$, rather than $4$. We have not found away to recover the last remaining form, $Z_{3|3}$, directly from an element of $\Spin_{24}(\CC)$, but we have 
\begin{gather}
	Z_{3|3}=2\phi_g-\frac23\phi_{g^3}-\frac13\phi_e
\end{gather}
for $g\in \Co_0$ with Frame shape $\pi_g=1^3.3^{-2}.9^3$ (take $D_g=9$ for the computation of $\phi_g$ here), which may be regarded as an analogue.

\subsection{Conway Modules}\label{app:k3teg-cms}

Both \cite{Duncan07} and \cite{conway} are concerned with moonshine for the Conway group, but the former focusses on $\vfn$, whereas the latter puts a spotlight on $\vsn$. As explained in \cite{conway}, these two objects are isomorphic as VOSAs, but inequivalent as representations of $\Co_0$. Indeed, the action of $\Co_0$ on $\vsn$ is faithful, whereas the action of $\Co_0$ on $\vfn$ factors through its center to the (sporadic) simple group $\Co_1:=\Co_0/\langle \gz\rangle$ (cf. e.g. \cite{ATLAS}). Here $\gz$ denotes the unique non-trivial central element of $\Co_0$, which is realized by $-I$ as an automorphism of $\LL$ (cf. (\ref{eqn:Co0AutLL})). 

To make our discussion explicit and concrete let $A$ denote the VOSA of $24$ free fermions, and let $A_\tw$ be an irreducible canonically twisted module for $A$. Then ${\cal A}:=A\oplus A_\tw$ admits a structure $({\cal A}, Y,\omega, {\bf v})$ of intertwining operator algebra, and the spin group $\Spin_{24}(\CC)$ is the automorphism group of this structure. Now according to the construction of \cite{Duncan07} there exists a vector $\tau\in A_\tw$ with the property that if $Y(\tau,z)=\sum_{n\in \frac12\ZZ}\tau_nz^{-n-1}$ then the operators $\tau_{n}$ for $n\in \frac12\ZZ$ generate actions of the Neveu--Schwarz and Ramond Lie superalgebras (cf. e.g. \cite{beast}) on ${\cal A}$. Thus it is natural to consider the subgroup of $\Spin_{24}(\CC)=\Aut({\cal A})$ that fixes $\tau$. It follows from the results of \cite{Duncan07} that this fixing group is none other than the Conway group,  $\Co_0$.

Now let $A=A^0\oplus A^1$ and $A_\tw=A_\tw^0\oplus A_\tw^1$ be the eigenspace decompositions for the action of the central element $\gz\in \Co_0$ on $A$ and $A_\tw$, respectively, so that $\gz$ acts as $(- I)^k$ on $A^k\oplus A_\tw^k$ for $k\in \{0,1\}$. Then the intertwining operator algebra (IOA) structure on ${\cal A}$ restricts to VOSA structures on $A^0\oplus A_\tw^0$ and $A^0\oplus A_\tw^1$, and the distinguished vector $\tau$ lies in $A_\tw^0$, and generates a representation of the Neveu--Schwarz Lie superalgebra on $A^0\oplus A_\tw^0$. 

Now as VOSAs with $\Co_0$-module structure we have $\vfn=A^0\oplus A^0_\tw$ and $\vsn=A^0\oplus A^1_\tw$. 
Both VOSAs admit (non-faithful) actions of $\Spin_{24}(\CC)$ by automorphisms, but we can naturally isolate an action of the Conway group in the case of $\vfn$ as follows.
Recall that an {\em ${\cal N}=1$ structure} on a VOSA $V$ is a choice of vector in $V$ for which the modes of the corresponding vertex operator generate a representation of the Neveu--Schwarz superalgebra on $V$. Then, according to the discussion above, $\tau$ defines an ${\cal N}=1$ structure on $\vfn$, and $\Co_1=\Co_0/\langle \gz\rangle$ is the subgroup of $\Aut(\vfn)$ that preserves this structure.

How about for $\vsn$? Well, it is no less natural to consider the subgroup of $\Aut(\vsn)$ that fixes $\tau$, which is precisely $\Co_0$. Since $\tau$ does not belong to $\vsn$ it does not define an ${\cal N}=1$ structure on $\vsn$ in the sense of \cite{Duncan07}, but it does belong to the canonically twisted $\vsn$-module $\vsn_\tw=A^1\oplus A_\tw^0$, and, according to our discussion, the modes of suitable corresponding intertwining operators generate representations of the Neveu--Schwarz and Ramond superalgebras on $\vsn\oplus \vsn_\tw$. With this in mind we make the following definition. For $V$ a VOSA define a {\em Ramond sector ${\cal N}=1$ structure} for $V$ to be a choice of vector $\tau\in V_\tw$, for a canonically twisted $V$-module $V_\tw$, with the property that the modes attached to $\tau$ by some intertwining operator on $V\oplus V_\tw$ generate representations of the Neveu--Schwarz and Ramon superalgebras on $V\oplus V_\tw$. Then we have shown that $\tau$ defines a Ramond sector ${\cal N}=1$ structure for $\vsn$, and $\Co_0$ arises as the automorphism group of this structure.

Finally we comment on the question of what happens when we take $\vfn$ in place of $\vsn$ in the setup of \cite{derived}. The question makes sense because the construction of \S9 of op. cit. applies equally well to $\vfn$ as it does to $\vsn$, but actually there is no difference in the Jacobi forms that one obtains. This is because if $G$ is any subgroup of $\Co_0=\Aut(\LL)$, or the orthogonal group $\SO(\LL\otimes_\ZZ\CC)=\SO_{24}(\CC)$ for that matter, that fixes a vector in $\LL\otimes_\ZZ\RR$, then it fixes an orthonormal vector $v$ in the space $\LL\otimes_\ZZ\CC$, which is naturally identified with $A^1_{\frac12}$. Now the zero mode $v(0)$ of the associated vertex operator $A\otimes A_\tw\to A_\tw((z^{\frac12}))$ defines an isomorphism of $G$-modules $A^0_\tw \to A^1_\tw$, since $G$ fixes $v$ by assumption. So $\vfn_\tw=A^1\oplus A_\tw^1$ and $\vsn_\tw=A^1\oplus A_\tw^0$ are the same as $G$-modules, and so the twinings of the K3 elliptic genus that we can recover from $\vfn_\tw$ and $\vsn_\tw$ coincide.


\begin{thebibliography}{99}

\bibitem{mukai} S. Mukai, ``Finite groups of automorphisms of $K3$ surfaces and the Mathieu group,'' {\em Inventiones Mathematicae}, {\bf 94}(1), 183--221 (1998).

\bibitem{eot} T. Eguchi, H. Ooguri, Y. Tachikawa, ``Notes on the $K3$ Surface and the Mathieu group $M_{24}$,'' {\em Experimental Mathematics}, {\bf 20}, 91--96 (2011).
\href{https://arxiv.org/pdf/1004.0956.pdf}{{\tt arXiv:1004.0956}}.

\bibitem{eguchi1} T. Eguchi, A. Taormina, ``Unitary representations of the $N=4$ superconformal algebra,'' {\em Physics Letters B}, {\bf 196}(1), 75--81 (1987).

\bibitem{eguchi2} T. Eguchi, A. Taormina, ``Character formulas for the $N=4$ superconformal algebra,'' {\em Physics Letters B}, {\bf 200}(3), 315--322 (1988).

\bibitem{eguchi3} T. Eguchi, A. Taormina, ``On the unitary representations of $N=2$ and $N=4$ superconformal algebras,'' {\em Physics Letters B}, {\bf 210}(1-2), 125--132 (1988).



\bibitem{dyons} M. Cheng, ``$K3$ surfaces, $\mathcal{N}= 4$ dyons and the Mathieu group $ M_{24}$,'' {\em Communications in Number Theory and Physics}, {\bf 4}(4), 623--657 (2010).
\href{https://arxiv.org/pdf/1005.5415.pdf}{{\tt arXiv:1005.5415}}.



\bibitem{gaberdiel_math1} M. Gaberdiel, S. Hohenegger, R. Volpato, ``Mathieu twining characters for $K3$,'' {\em Journal of High Energy Physics}, (9), 058, 20 (2010).
\href{https://arxiv.org/pdf/1006.0221.pdf}{{\tt arXiv:1006.0221}}.

\bibitem{gaberdiel_math2} M. Gaberdiel, S. Hohenegger, R. Volpato, ``Mathieu Moonshine in the elliptic genus of $K3$,'' {\em Journal of High Energy Physics}, (10), 062, 24 (2010).
\href{https://arxiv.org/pdf/1008.3778.pdf}{{\tt arXiv:1008.3778}}.

\bibitem{eguchi4} T. Eguchi, K. Hikami, ``Note on twisted elliptic genus of $K3$ surface,'' {\em Physics Letters B}, {\bf 694}(4-5), 446--455 (2011).
\href{https://arxiv.org/pdf/1008.4924.pdf}{{\tt arXiv:1008.4924}}.

\bibitem{surf1} A. Taormina, K. Wendland, ``The overarching finite symmetry group of Kummer surfaces in the Mathieu group $M_{24}$,'' {\em Journal of High Energy Physics}, (8), 125 (2013).
\href{https://arxiv.org/pdf/1107.3834.pdf}{{\tt arXiv:1107.3834}}.

\bibitem{umbral1} M. Cheng, J. Duncan, J. Harvey, ``Umbral moonshine,'' {\em Communications in Number Theory and Physics}, {\bf 8}(2), 101--242 (2014).
\href{https://arxiv.org/pdf/1204.2779.pdf}{{\tt arXiv:1204.2779}}.

\bibitem{umbral2} M. Cheng, J. Duncan, J. Harvey, ``Umbral moonshine and the Niemeier lattices,'' {\em Research in the Mathematical Sciences}, {\bf 1}(3), 1--81 (2014).
\href{https://arxiv.org/pdf/1307.5793.pdf}{{\tt arXiv:1307.5793}}.


\bibitem{surf2} A. Taormina, K. Wendland, ``A twist in the $M_{24}$ Moonshine story,'' {\em Confluentes Mathematici}, {\bf 7}(1), 83--113 (2015).
\href{https://arxiv.org/pdf/1303.3221.pdf}{{\tt arXiv:1303.3221}}.

\bibitem{surf3} M. Gaberdiel, C. Keller, H. Paul, ``Mathieu moonshine and symmetry surfing,'' {\em Journal of Physics A: Mathematical and Theoretical}, {\bf 50}(47), 474002 (2017).
\href{https://arxiv.org/pdf/1609.09302.pdf}{{\tt arXiv:1609.09302}}.

\bibitem{LG} M. Cheng, F. Ferrari, S. Harrison, N. Paquette, ``Landau-Ginzburg orbifolds and symmetries of $K3$ CFTs,'' {\em Journal of High Energy Physics}, (1), 046 (2017).
\href{https://arxiv.org/pdf/1512.04942.pdf}{{\tt arXiv:1512.04942}}.

\bibitem{HM2020} J. Harvey, G. Moore, ``Moonshine, Superconformal Symmetry, and Quantum Error Correction,'' {\em Journal of High Energy Physics}, (5), 146 (2020).
\href{https://arxiv.org/pdf/2003.13700.pdf}{{\tt arXiv:2003.13700}}.

\bibitem{gaberdiel_k3} M. Gaberdiel, S. Hohenegger, R. Volpato, ``Symmetries of $K3$ sigma models,'' {\em Communications in Number Theory and Physics}, {\bf 6}(1), 1--50 (2012).
\href{https://arxiv.org/pdf/1106.4315.pdf}{{\tt arXiv:1106.4315}}.

\bibitem{huybrecht} D. Huybrechts, ``On derived categories of $K3$ surfaces, symplectic automorphisms and the Conway group,'' in {\em Development of Moduli Theory---Kyoto 2013}. Advanced Studies in Pure Mathematics, vol. 69 (Mathematical Society of Japan, Tokyo, 2016) 387--405.
\href{https://arxiv.org/pdf/1309.65280.pdf}{{\tt arXiv:1309.6528}}.

\bibitem{derived} J. Duncan, S. Mack-Crane, ``Derived equivalences of $K3$ surfaces and twined elliptic genera,'' {\em Research in the Mathematical Sciences}, {\bf 3}, Art. 1, 47  (2016).
\href{https://arxiv.org/pdf/1506.06198.pdf}{{\tt arXiv:1506.06198}}.

\bibitem{umbral_k3} M. Cheng, S. Harrison, ``Umbral moonshine and $K3$ surfaces.'' {\em Communications in Mathematical Physics}, {\bf 339}(1), 221--261 (2015).
\href{https://arxiv.org/pdf/1406.0619.pdf}{{\tt arXiv:1406.0619}}.

\bibitem{k3_lattices} M. Cheng, S. Harrison, R. Volpato, M. Zimet, ``$K3$ string theory, lattices and moonshine,'' {\em Research in the Mathematical Sciences}, {\bf 5}(3), 32 (2018).
\href{https://arxiv.org/pdf/1612.04404.pdf}{{\tt arXiv:1612.04404}}.

\bibitem{Duncan07} J. Duncan, ``Super-Moonshine for Conway's Largest Sporadic Group,'' {\em Duke Mathematical Journal}, {\bf 139}(2), 255--315 (2007).
\href{https://arxiv.org/pdf/math/0502267.pdf}{{\tt arXiv:math/0502267}}.

\bibitem{conway} J. Duncan, S. Mack-Crane, ``The moonshine module for Conway's group,'' {\em Forum of Mathematics, Sigma}, {\bf 3}, e10, 52 (2015).
\href{https://arxiv.org/pdf/1409.3829.pdf}{{\tt arXiv:1409.3829}}.






\bibitem{k3refl} 
A. Taormina, K. Wendland, (2017). ``The Conway moonshine module is a reflected $K3$ theory,'' {\em Advances in Theoretical and Mathematical Physics}, {\bf 24}(1) (2020). 
\href{https://arxiv.org/pdf/1704.03813.pdf}{{\tt arXiv:1704.03813}}.

\bibitem{refl} 
T. Creutzig, J. Duncan, W. Riedler, ``Self-Dual Vertex Operator Superalgebras and Superconformal Field Theory,'' {\em Journal of Physics A}, {\bf 51}(3) 034001, 29 (2018). 
\href{https://arxiv.org/pdf/1704.03678.pdf}{{\tt arXiv:1704.03678}}.

\bibitem{FLMBerk}
I. Frenkel, J. Lepowsky, A. Meurman, ``A moonshine module for the {M}onster'', in {\em Vertex operators in mathematics and physics (Berkeley, Calif., 1983)}. Math. Sci. Res. Inst. Publ., vol. 3 (Springer, 1985) 231--273.


\bibitem{t4paper} 
R. Volpato, ``On symmetries of $\mathcal{N}=(4,4)$ sigma models on $T^4$,'' {\em Journal of High Energy Physics}, 08 (2014) 094.
\href{https://arxiv.org/pdf/1403.2410.pdf}{{\tt arXiv:1403.2410}}.


\bibitem{ATLAS}
J. Conway, R. Curtis, S. Norton, R. Parker, R. Wilson, {\em Atlas of finite groups}. (Clarendon Press, Oxford, 1985).

\bibitem{FLM} I. Frenkel, J. Lepowsky, A. Meurman, {\em Vertex operator algebras and the Monster}. Pure and Applied Mathematics, vol. 134 (Academic Press Inc., Boston, 1988).

\bibitem{DolanGoddardMontague} Dolan, L., Goddard, P., \& Montague, P. (1990). Conformal field theory of twisted vertex operators. Nuclear Physics B, 338(3), 529-601.



\bibitem{hiker} 
W. Nahm, K. Wendland, ``A Hiker's Guide to K3. Aspects of $N=(4, 4)$ Superconformal Field Theory with Central Charge $c= 6$," {\em Communications in Mathematical Physics}, {\bf 216}, 85--138 (2001).

\bibitem{lamshi}
C. Lam, H. Shimakura, ``Reverse orbifold construction and uniqueness of holomorphic vertex operator algebras,'' {\em Transactions of the American Mathematical Society}, {\bf 372}(10), 7001--7024 (2019).
\href{https://arxiv.org/pdf/1606.08979.pdf}{{\tt arXiv:1606.08979}}.




\bibitem{Borcherds86} Borcherds, R.R.. Vertex algebras, {K}ac-{M}oody algebras, and the {M}onster.
Proc. Nat. Acad. Sci. U.S.A.,
{\bf 83} 10
(1986).

\bibitem{Gross:1985fr}
Gross,~D.~J.,~Harvey, J.~A.,~ Martinec,E.~J.~, and Rohm, R. (1985).
Heterotic String Theory. 1. The Free Heterotic String.
Nucl. Phys. B (vol. 256).




\bibitem{Borcherds}
R.~E. Borcherds, ``Monstrous moonshine and monstrous lie superalgebras,'' {\em
	Inventiones mathematicae} {\bfseries 109} no.~1, (1992) 405--444.





\bibitem{holomvoas} 
J.~van Ekeren, S.~M{\"o}ller, N. Scheithauer, 
``Construction and classification of holomorphic vertex operator algebras," {\em Journal f{\"u}r die reine und angewandte Mathematik} (Crelle's Journal) (2017).
\href{http://arxiv.org/abs/1507.08142}{{\ttfamily arXiv:1507.08142}}. 



\bibitem{packings} 
J. Conway, N. Sloane, {\em Sphere packings, lattices and groups}. 
Grundlehren der Mathematischen Wissenschaften, vol. 290, 3rd edn. (Springer, New York, 1999).

\bibitem{matmooreview}
C. Cheng, J. Duncan, ``The Largest Mathieu Group and (Mock) Automorphic Forms,'' in {\em String-Math 2011}. Proceedings of Symposia in Pure Mathematics, vol. 85 (American Mathematical Society, Providence, RI, 2012), 53--82. 
\href{https://arxiv.org/pdf/1201.4140.pdf}{{\tt arXiv:1201.4140}}.

\bibitem{beast}
L. Dixon, P. Ginsparg, J. Harvey, ``Beauty and the beast: superconformal symmetry in a {M}onster module,'' {\em Communications in Mathematical Physics}, {\bf 119}(2), 221--241 (1988).


\end{thebibliography}
\end{document}